\documentclass[a4paper,UKenglish,cleveref, autoref, thm-restate]{lipics-v2021}

\usepackage{todonotes}
\usepackage{amsmath,amsthm}
\usepackage{soul}
\usepackage{color}
\usepackage{wrapfig}
\usepackage{amsthm}

\usepackage{comment}
\excludecomment{hidden}

\renewcommand{\subparagraph}[1]{\par\smallskip\noindent\textbf{{#1}}}

\theoremstyle{definition}
\newtheorem{branchrule}{Branching rule}

\nolinenumbers

\pdfoutput=1 %
\hideLIPIcs  %

\title{Split-or-decompose: Improved FPT branching algorithms for
maximum agreement forests} %

\titlerunning{Improved FPT branching algorithms for maximum agreement forests} %

\author{David Mestel}{Department of Advanced Computing Sciences, Maastricht University, The Netherlands }{david.mestel@maastrichtuniversity.nl}{}{}
\author{Steven Chaplick}{Department of Advanced Computing Sciences, Maastricht University, The Netherlands }{s.chaplick@maastrichtuniversity.nl}{https://orcid.org/0000-0002-1825-0097}{}
\author{Steven Kelk}{Department of Advanced Computing Sciences, Maastricht University, The Netherlands }{steven.kelk@maastrichtuniversity.nl}{}{}
\author{Ruben Meuwese}{Department of Advanced Computing Sciences, Maastricht University, The Netherlands }{r.meuwese@maastrichtuniversity.nl}{}{Supported by the Dutch Research Council (NWO) KLEIN 1 grant Deep kernelization for
phylogenetic discordance, project number OCENW.KLEIN.305.}

\authorrunning{D.~Mestel, S.~Chaplick, S.~Kelk, R.~Meuwese} %

\Copyright{Jane Open Access and Joan R. Public} %



\category{} %

\relatedversion{} %

\EventEditors{John Q. Open and Joan R. Access}
\EventNoEds{2}
\EventLongTitle{42nd Conference on Very Important Topics (CVIT 2016)}
\EventShortTitle{CVIT 2016}
\EventAcronym{CVIT}
\EventYear{2016}
\EventDate{December 24--27, 2016}
\EventLocation{Little Whinging, United Kingdom}
\EventLogo{}
\SeriesVolume{42}
\ArticleNo{23}

\newcommand{\magicunrootednumber}{2.846}
\newcommand{\bestrootednumber}{2.3391}

\newcommand{\qfalse}{\bot}

\newcommand{\uMAF}{{\rm uMAF}}
\newcommand{\rMAF}{{\rm rMAF}}

\newcommand{\SK}[1]{{\textcolor{black}{#1}}}
\newcommand{\DM}[1]{{\textcolor{black}{#1}}}

\begin{document}

\maketitle

\begin{abstract}
Phylogenetic trees are leaf-labelled trees used to model the evolution of species. In practice it is not uncommon to obtain two topologically distinct trees for the same set of species, and this motivates the use of distance measures to quantify dissimilarity. A well-known measure is the maximum agreement forest (MAF).
Computing such a MAF is NP-hard and so considerable effort has been invested in finding FPT algorithms, parameterised by $k$, the number of components of a MAF.
In this work we present improved algorithms for both the unrooted and rooted cases, with runtimes $O^*(\magicunrootednumber^k)$ and $O^*(\bestrootednumber^k)$ respectively.  
The key to our improvement is a novel branching strategy in which we show that any overlapping components obtained on the way to a MAF can be `split' by a branching rule with favourable branching factor, and then the problem can be decomposed into disjoint subproblems to be solved separately. 
\end{abstract}

\clearpage

\section{Introduction}
A phylogenetic tree is a binary tree whose leaves are bijectively labeled with a set of $n$ species (or,
more generically, a set of \emph{taxa}) $X$. These trees are common in the systematic study
of evolution: the leaves represent contemporary species and the internal vertices of the tree
represent hypothetical common ancestors. There are many different methods for constructing such trees and it is not unusual for different methods to produce topologically distinct trees on the same input data, or the same method to produce distinct trees depending on which part of the genome has been used as input data. This naturally leads to the challenge of quantifying the topological dissimilarity of two trees on the same set of taxa $X$.

Here we consider the \emph{maximum agreement forest} (MAF) criterion for quantifying dissimilarity. Informally, this is a minimum-size partition of %
$X$ which splits both trees into (up to homeomorphism) the same set of disjoint, leaf-labelled subtrees. The agreement forest model is useful when the trees are presumed to differ due to the influence of complex biological phenomena, such as horizontal gene transfer, which topologically rearrange the underlying evolutionary history \cite{WhiddenZB14}. 
Unfortunately, it is NP-hard (in fact, APX-hard) to compute a MAF, both in the rooted variant of the problem (rMAF) \cite{bordewich2007computing}, when the two input trees have a root denoting the direction of evolution, and in the unrooted variant, where the input trees are undirected (uMAF) \cite{HeinJWZ96,kelk2024agreement}. These variants have been intensively studied in the last decade by the algorithms community; see \cite{WhiddenBZ13,shi2018parameterized,BulteauW19} for good overviews.
The fixed parameter tractable (FPT) status of rMAF and uMAF dates back
 to 2004 and 2001, respectively, when linear kernels were established \cite{BordewichS04,AllenS01}. Regarding FPT branching algorithms, which are the focus of this article, an $O^{*}(4^k)$ algorithm was given for rMAF in 2008 \cite{bordewich20083}, followed by $O^{*}(3^k)$ in 2009 \cite{whidden2009unifying}, $O^{*}(2.42^k)$ in 2013\footnote{A further advancement was claimed in the PhD thesis of Whidden (2013) but this remains unpublished.} \cite{WhiddenBZ13} and finally $O^{*}(2.3431^k)$ in 2015 \cite{chen2015faster}, by an elaborate case analysis of at least 23 cases. Here $k$ refers to the size of the optimum.  For uMAF, an $O^{*}(4^k)$ algorithm was given in 2007 by \cite{hallett2007faster} and (via a rather different technique) again in 2009 \cite{whidden2009unifying}, and finally $O^{*}(3^k)$ was given in 2015 \cite{ChenFS15} which builds on techniques from \cite{WhiddenBZ13}.

\subparagraph{Our Contribution.} In this work we give FPT branching algorithms with running times $O^{*}(\bestrootednumber^{k})$  and $O^{*}(\magicunrootednumber^{k})$ 
for rMAF and uMAF, respectively, thus improving on previous results: quite significantly in the case of uMAF. The key insight is a novel `splitting' branching rule, which ensures that the components created on the way to constructing the MAF can be assumed to be disjoint.
More precisely, existing algorithms
proceed by successively cutting edges in one of the two trees until a MAF is reached, so maintaining at each step of the algorithm a pair $(T,F')$ of one of our original trees $T$ and a forest $F'$ formed by cutting edges in the other input tree $T'$. The algorithm then branches on a
set of further cuts to be made in the forest 
(according to some branching rule that must be proved to be safe, i.e. it must be shown that any MAF extending the current configuration must contain one of the candidate sets of cuts), and continues.  The number of branches per cut edge controls the running time of the algorithm, and is called the `branching factor'.

Our new insight is to show that
whenever
the components of our forest overlap in the other tree (i.e. the subtrees of $T$ induced by the taxa of each component of $F'$ are not disjoint) then there is a branching rule we can apply with branching factor at most 2.  We can apply this `splitting' branching rule until the components of $F'$ do not overlap in $T$, and then solve each component
as a fresh instance of the MAF problem.  
Since the running time of all these algorithms is exponential in the number of cuts $k$, breaking an instance into multiple subproblems leads to a major gain as long as the remaining cuts are not all in one subproblem. Most of the technical work for our algorithms will be in showing that this does not occur, or rather that if it would occur then there are alternative branching rules we can apply instead with favourable branching factors.

We are hopeful that this new technique, which limits the need for extensive case analysis, could lead to improvements more widely on the many variants of MAF that exist; it is arguably the first structural methodological advance in MAF branching algorithms since the foundational work of Whidden et al.~15 years ago \cite{whidden2009unifying,WhiddenBZ13}. Although the running time of $O^{*}(\bestrootednumber^{k})$ for rMAF is a small improvement over the previous best of $O^{*}(2.3431^{k})$, we are optimistic that in the future our technique can lead to further improvements. Indeed, as another proof-of-concept we improve the fastest uMAF branching algorithm on two caterpillars (which are path-like phylogenetic trees) from $O^{*}(2.49^{k})$ to
$O^{*}(2.4634^{k})$ using our technique.
Although empirical evaluation is outside the scope of this paper, we expect that overlapping components will arise frequently in practice, so that adding our split-or-decompose technique may improve the performance of existing algorithms on practical problems.

In Section \ref{sec:preliminaries} we set out some preliminary definitions. In Section \ref{sec:existing} we describe the existing FPT branching algorithms for uMAF and rMAF which we build upon. Section \ref{sec:splitordecompose} introduces our `split-or-decompose' technique. This is then applied to uMAF and rMAF in Sections \ref{sec:unrooted} and \ref{sec:rooted} respectively. Finally in Section \ref{sec:future} we discuss some future research directions.

\section{Preliminaries}\label{sec:preliminaries}
Throughout this paper, $X$ denotes a finite set of \emph{taxa}. We consider both unrooted and rooted binary phylogenetic trees. We introduce terminology and notation for unrooted trees, and then explain where the differences lie for rooted trees.  Figure~\ref{fig:examples} illustrates some key~concepts.

\begin{figure}
\centering
\includegraphics{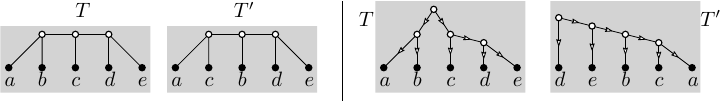}
\caption{Left: two unrooted binary phylogenetic trees on taxa $X=\{a,b,c,d,e\}$. An uMAF for these
two trees has 2 components, e.g. $\{\{a,b,c\}, \{d,e\}\}$. Right: two rooted binary phylogenetic trees on the same set of taxa. An rMAF for these two trees has 3 components, e.g. $\{\{a,c\}, \{b\}, \{d,e\}\}$.}
\label{fig:examples}
\end{figure}

\subparagraph{Unrooted trees.}
An {\it unrooted binary phylogenetic tree} $T$ on $X$ is a simple undirected tree whose leaves are bijectively labeled with $X$ and whose other vertices all have degree 3. For simplicity we refer to an unrooted binary phylogenetic tree as a {\it phylogenetic tree} or even just a {\it tree} when it is clear from the context. For two phylogenetic trees $T, T'$ on the same set of taxa $X$ we write $T=T'$ if there is an isomorphism between $T$ and $T'$ that is identity on $X$. Note that, if $|X|\leq 3$, $T=T'$ always holds. A tree is a \emph{caterpillar} if deleting all leaves yields a path. Two leaves, say $a$ and $b$, of $T$ are called a {\it cherry} $\{a,b\}$ of $T$ if they are adjacent to a common vertex. 
For $X' \subseteq X$, we write $T[X']$ to denote the unique, minimal subtree of $T$ that connects all elements in $X'$. For brevity we call $T[X']$ the \emph{embedding} of  (the subtree induced by) $X'$ in $T$.
We refer to the phylogenetic tree on $X'$ obtained from $T[X']$ by suppressing degree-2 vertices as the {\it restriction of $T$ to $X'$}  and we denote this by $T|X'$. If $X' \subseteq X$ and $T|X' = T'|X'$, we
say that the subtrees induced by $X'$ in $T$ and $T'$ are \emph{homeomorphic}; for brevity we often just say ``$X'$ is homeomorphic''. For a phylogenetic tree $T$ we will sometimes write $\ell(T)$ to denote the taxa at its leaves. The meaning of $\ell(.)$ extends as expected also to embeddings of phylogenetic trees.

A subtree $T^{*}$ of a phylogenetic tree $T$ is \emph{pendant} if it can be detached from $T$ by deleting a single edge of $T$. For simplicity we often refer to a pendant subtree by the set of taxa $X^{*} \subset X$ at its leaves. 
Let $X^{*} \subset X$ and $T, T'$ be two phylogenetic trees on $X$, such that $X^{*}$ induces a pendant subtree in both $T$ and $T'$. 
We say that the pendant subtree $T'$ is \emph{common ignoring rooting} if $T|X^{*} = T'|X^{*}$. 
The pendant subtree $T'$ is \emph{common including rooting} if $T|(X^{*} \cup \{y\}) = T'|(X^{*} \cup \{y\})$ where $y$ is any taxon in $X \setminus X^{*}$. The second notion of common pendancy is stronger than the first; informally, it says that the two pendant subtrees not only have the same topology, but also attach to their parent trees at the same place. For the two unrooted trees $T, T'$ in Figure \ref{fig:examples}, taking $X^{*}=\{a,b,c\}$ yields a common pendant subtree ignoring rooting, but not including rooting.

\subparagraph{Rooted trees.}
A \emph{rooted} binary phylogenetic tree on $X$ is a simple directed  tree whose leaves are bijectively labelled by $X$ and all internal nodes have indegree-1 and outdegree-2, except (when $|X| \geq 2$) a single node with indegree-0 and outdegree-2 (the \emph{root}). 
As in the unrooted case, we often simply write \emph{phylogenetic tree} or \emph{tree} to denote a rooted binary phylogenetic tree. 
The definitions for rooted trees go through as for unrooted trees, with the following exceptions. For $|X| \leq 3$, $T=T'$ does not always hold (but it does if $|X|\leq 2)$. For $X' \subseteq X$, we write $T|X'$ to denote the rooted  tree obtained from $T[X']$ by suppressing all nodes that have indegree-1 and outdegree-1. 
Also, unlike unrooted trees the notion of a common pendant subtree requires no nuancing: the location of the root in the subtree is always taken into account. 
Note that a rooted tree (when $|X|\geq 2$) might only have one cherry, and such a tree is a caterpillar. However, an unrooted tree $T$ (when $|X|\geq 4$) always has at least two cherries $\{a,b\}$, $\{c,d\}$ where $a,b,c,d$ are distinct and $T[\{a,b\}]$ is disjoint from $T[\{c,d\}]$. The only unrooted trees with at most two cherries are caterpillars.
Finally, when referring to a path in a rooted tree, we mean a path in the underlying unrooted tree.

\subparagraph{Agreement forests.}
The following definitions hold for both rooted and unrooted trees. Let $T$ and $T'$ be two phylogenetic trees  on $X$. Let $F = \{B_1,B_2,\ldots,B_k\}$ 
be a partition of $X$, where each block $B_i$ with $i\in\{1,2,\ldots,k\}$ is  referred to as a \emph{component} of $F$. We say that $F$ is an \emph{agreement forest} for $T$ and $T'$ if the following conditions hold.
\begin{enumerate}
\item [(1)] For each $i\in\{1,2,\ldots,k\}$, we have $T|B_i = T'|B_i$.
\item [(2)] For each pair $i,j\in\{1,2,\ldots,k\}$ with $i \neq j$, we have that
$T[B_i]$ and $T[B_j]$ are vertex-disjoint in $T$, and $T'[B_i]$ and $T'[B_j]$ are vertex-disjoint in $T'$.
\end{enumerate}
\noindent
Let $F=\{B_1,B_2,\ldots,B_k\}$ be an agreement forest for $T$ and $T'$. The \emph{size} of $F$ is simply its number of components; i.e. $k$. Moreover, an agreement forest with the minimum number of components (over all agreement forests for $T$ and $T'$) is called a \emph{maximum agreement forest (MAF)} for $T$ and $T'$.
The number of components of a maximum agreement forest for $T$ and $T'$ is denoted by $d_\uMAF(T,T')$, when $T$ and $T'$ are unrooted trees, and $d_\rMAF(T,T')$ when $T$ and
$T'$ are rooted. Although they share much common structure, there are subtle combinatorial differences between the rMAF and uMAF models, which explains the differences in running times obtained thus far for the two models.\footnote{Also, with rMAF, 
an extra taxon $\rho$ pendant to the root in both trees is commonly added.
This is done to ensure that
rMAF correctly models the \emph{rooted subtree prune and regraft} distance. 
We do not adopt this convention here, but just note that our rMAF algorithm works with or without $\rho$.}

It is well-known that, when computing $d_\uMAF(T,T')$, pendant subtrees that are common including rooting can be reduced to a single taxon without changing $d_\uMAF(T,T')$ \cite{AllenS01}. When computing $d_\rMAF(T,T')$ common pendant subtrees can also be reduced to a single taxon without changing $d_\rMAF(T,T')$ \cite{BordewichS04}. This is why, for both rMAF and uMAF, if $T$ and $T'$ both contain a cherry $\{a,b\}$, this common cherry can always be reduced to a single taxon.

\section{Existing branching strategies}
\label{sec:existing}
We first describe the branching algorithm of Chen et al. \cite{ChenFS15} for uMAF, with running time  $O^{*}(3^k)$, which was based on earlier work by Whidden et al. \cite{WhiddenBZ13}. Then, for rMAF, we explain the algorithm from~\cite{WhiddenBZ13} with a running time of $O^{*}(2.42^k)$. 
These two classical branching schemes form the starting point for our new~algorithms.

\subsection{Chen et al.'s algorithm for uMAF \cite{ChenFS15}}
We start with two unrooted trees $T$ and $T'$ on $X$. 
The high-level idea is to progressively cut edges in one of the trees, say $T'$, to obtain a forest $F'$ (initially $F'=T'$) with an increasing number of components, until it becomes an \emph{agreement} forest for $T$ and $T'$. 
Each edge cut increases the size of the forest by 1. 
Hence, to determine whether there exists an agreement forest with at most $k$ components, we can make at most $k-1$ edge cuts in $T'$. At each step various ``tidying up'' operations are applied to exhaustion:
\begin{enumerate}
\item If a singleton component is created in $F'$ (i.e. a component comprising a single taxon $a$), then we delete $a$ from both $T$ and $F'$;
\item Degree-2 nodes are always suppressed;
\item Common cherries are always reduced into a single taxon in both trees. For example by discarding an arbitrary taxon from the common cheery in both $T$ and $F'$. 
\end{enumerate}
The way we will choose edges to cut, combined with the tidying up operations, ensures that (unlike $T'$) $T$ remains connected at every step. More formally: if, at a given step, $T'$ has been cut into a forest $F'$ and $X'$ is the union of the taxa in $F'$, then $T$ will have been transformed into $T|X'$. To keep notation light we will henceforth refer simply to $T$ and $F'$, with the understanding that we are actually referring to the tree-forest pair encountered at a specific iteration of the algorithm.
The general decision problem is as follows. We are given a $(T,F')$ tree-forest pair on a set $X' \subseteq X$ of taxa, and a parameter $k$, and we wish to know: is it possible to transform $F'$ into an agreement forest making at most $k$ further cuts (more formally: an agreement forest for $T|X'$ and $T'|X'$)? If we can answer this query we will, in particular, be able to answer the \emph{original} query of whether the original input  $(T, T')$ has an agreement forest with at most $k+1$ components. Note that for $(T, T')$ the query $k=0$ returns TRUE if and only if $T = T'$.
We use a recursive branching algorithm to answer this query, and \DM{write $T(k)$
to denote the number of final configurations at the lowest level of the branching process (i.e. the number of `leaf' calls which return a value without any further recursive calls) that
results from answering this query.  Note that the total number of recursive calls in this branching process is at most twice the number of final calls, and the total running time of the algorithm is at most a polynomial factor larger than the total number of calls (because inside each call only a polynomial amount of work will be undertaken).} For brevity, and because we are only interested in the exponential dependence of the running time on $k$, we henceforth simply refer to $T(k)$ as ``running time''.

The following branching rules are used; we assume that common cherries have already been (exhaustively) collapsed.

\begin{figure}[t]
\centering
\includegraphics{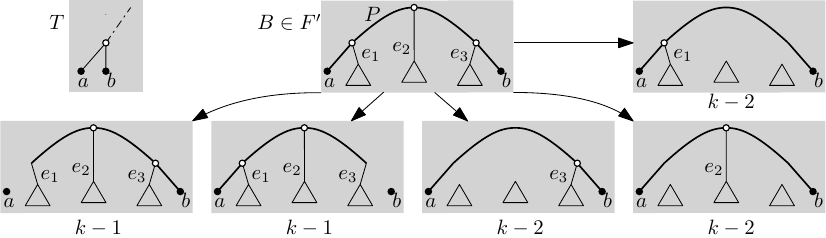}
\caption{An illustration of Chen's branching strategy for \emph{unrooted} agreement forests when $\{a,b\}$ is a cherry in $T$ and $a$ and $b$ are together in a component $B \in F'$. Here $t=3$ and the resulting recurrence is $2T(k-1)+3T(k-2)$.}
\label{fig:chent3}
\end{figure}

\begin{enumerate}
\item (DIFFERENT COMPONENTS) If $T$ has a cherry $\{a,b\}$ such that $a$ and $b$ are in different components of $F'$, then cut $a$ off, or cut $b$ off. This yields the inequality $T(k) \leq 2T(k-1)$ and a bound on the exponent of $2^k$.
\item (CHEN) If $T$ has a cherry $\{a,b\}$ such that $a$ and $b$ are in the same component $B$ of $F'$, then let $P$ be the unique path from $a$ to $b$ in $B$. Let $e_1, ..., e_t$ be the edges with exactly one endpoint on $P$.   Note that $t \geq 2$ otherwise $\{a,b\}$ is a common cherry. There are three possibilities: $a$ is cut off, $b$ is cut off, all but one of the $e_1, \ldots, e_t$ edges is cut off.
See Figure \ref{fig:chent3} for an illustration of how the branching works when $t=3$. This gives a recurrence $T(k) \leq 2T(k-1) + tT(k-(t-1))$.  If $t \geq 3$, then this recursion can be bounded by $3^k$. However, for $t=2$ the recursion yields $4^k$. To circumvent this, Chen et al. observe that the branch ``cut $b$ off'' can be dropped when $t=2$. This yields a recurrence $T(k-1) + 2T(k-1)$, producing a bound of $3^k$ for that case, and thus at most $3^k$ overall.
\end{enumerate}

\subsection{Whidden et al.'s algorithm for rMAF \cite{WhiddenBZ13}}
This is very similar to the above\footnote{The tidying up steps are also analogous, with the exception that we do not suppress nodes with indegree-0 and outdegree-2, i.e. the roots of components.} and, indeed, predates it. The second branching rule is slightly different as a result of the trees
being rooted.

\begin{figure}[tb]
\centering
\includegraphics{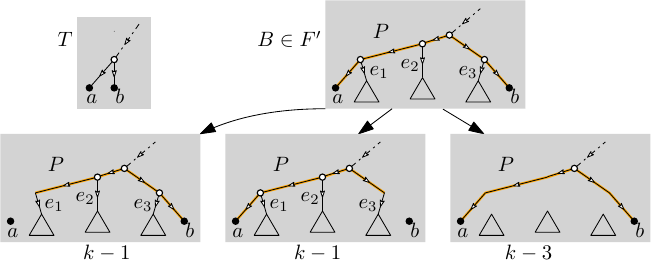}
\caption{An illustration of Whidden's branching strategy for \emph{rooted} agreement forests when $\{a,b\}$ is a cherry in $T$ and $a$ and $b$ are together in a component $B \in F'$. Here $t=3$, and the resulting recurrence is $2T(k-1)+T(k-3)$.}
\label{fig:whiddent3}
\end{figure}

\begin{enumerate}
\item (DIFFERENT COMPONENTS) 
This is the same as in the previous section.
\item (WHIDDEN) If $T$ has a cherry $\{a,b\}$ such that $a$ and $b$ are in the same component $B$ of $F'$, then let $P$ be the unique path from $a$ to $b$ in $B$. Let $e_1, ..., e_t$ be the edges whose tails, but not heads, lie on $P$.   Note that $t \geq 1$ otherwise $\{a,b\}$ is a common cherry. There are three possibilities: $a$ is cut off, $b$ is cut off, \emph{all} the $e_1, \ldots, e_t$ edges are cut off. See Figure \ref{fig:whiddent3}. This gives a recurrence $T(k) \leq 2T(k-1) + T(k-t)$. If $t \geq 2$, then this recursion is bounded by $(1 + \sqrt{2})^k \approx 2.42^k$. The case $t=1$ remains a bottleneck however, yielding $3^k$. However, when $t=1$ we do not need to consider the case when $a$ is cut off, or when $b$ is cut off: it is safe to simply deterministically cut the edge $e_1$ (Case 6.2 of \cite{WhiddenBZ13}). This is sufficient to obtain a bound on the recursion of $2.42^k$ for all $t \geq 1$. The limiting case is $t=2$, since the running time further decreases for $t \geq 3$. For $t=3$ we obtain approximately $2.206^k$, for example.
\end{enumerate}

\section{Split-or-decompose}
\label{sec:splitordecompose}

In this section
we describe our `split-or-decompose' technique, which applies to both rooted and unrooted trees.
First, in Section \ref{subsec:killoverlap} we describe a branching rule (branching rule \ref{rule:overlap}) with branching factor two that applies whenever we have overlapping components.  If this rule does not apply (or after it has been applied to exhaustion), we are guaranteed to have disjoint components which we can then solve as independent sub-problems.  This has two benefits. 
Firstly, at each subsequent level of our recursive algorithm we have two trees rather than a tree and a forest.  Secondly, if the remaining cuts are to any degree 
divided across
these subproblems then we get a significant gain (essentially because in general $\rho^{k_1}+\rho^{k_2}\ll \rho^{k_1+k_2}$).  In Section \ref{subsec:decompose} we describe the `recursion rule' used to achieve this.

Unlike earlier work, we frame our recursive algorithm as solving a hybrid decision-optimization problem. Formally: given $(T,F')$ the query asks, is it possible to reach an agreement forest with at most $k$ further cuts in $F'$, and if so what is the minimum number of cuts required to reach an agreement forest from this point?
This requires slightly different post-processing of the results for alternative branches inside branching rules: essentially, taking the minimum rather than the logical OR.  The formalism of an optimization problem is necessary in order to use the number of cuts required by one sub-problem to inform the depth to which we solve another.

\subsection{Splitting overlapping components}
\label{subsec:killoverlap}
The essence of our splitting strategy is to exploit the fact that at most one component of an agreement forest can use a given edge of a tree,
and that all other components need to respect the bipartition of their taxa induced by that edge. Before concretizing this notion we need to introduce some further concepts.

\begin{figure}[t]
    \centering
    \includegraphics[page=2]{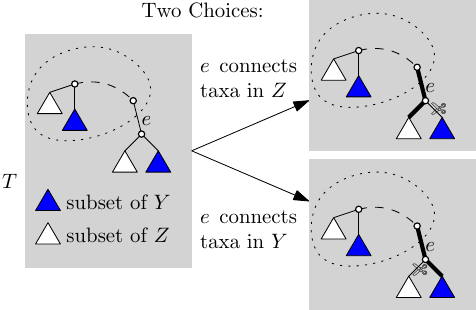}
    \caption{The \emph{branching} involved in splitting a tree $T$ into components according to a bipartition $(Y,Z)$ of the leaves of $T$. In any minimal split, the edge $e$ will be allocated either to a subtree whose leaves belong to $Y$ or to a subtree whose leaves belong to $Z$. Repeatedly applying this procedure allows a splitting core to be constructed (see Figure \ref{fig:split-def}).}
    \label{fig:split-branching}
\end{figure}

\begin{figure}[b]
    \centering
    \includegraphics[page=1]{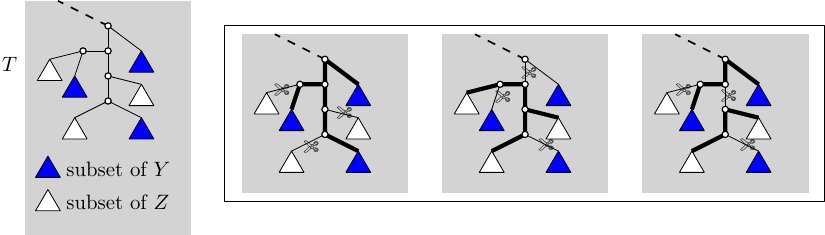}
    \caption{(left) A subtree $T'$ of a tree $T$ where the leaves (taxa) have been bipartitioned as $(Y,Z)$. (right) Three minimal cuts in $T'$ where no  resulting component has leaves from both of $Y$ and $Z$. These three cuts form a splitting core of the subtree $T'$ with respect to $(Y,Z)$. The cuts can be obtained by repeatedly identifying a pair of edges, at least one of which must be cut, and branching on which of the two is cut (see Figure \ref{fig:split-branching}).  }
    \label{fig:split-def}
\end{figure}

Let $T$ be a binary phylogenetic tree, and let $(Y,Z)$ be a bipartition of the labels of $T$.  Say that a cut $K \subseteq E(T)$ \emph{splits} $(Y,Z)$ if no component of $T\setminus K$ contains both $Y$- and $Z$-vertices.
Say that cut $K_1$ \emph{refines} cut $K_2$ if whenever two labels are connected in $T\setminus K_1$ they are connected in $T\setminus K_2$.  Say that a set $C=\{K_1,\ldots,K_n\}$ of non-empty cuts is a \emph{splitting core} for $(Y,Z)$ if for every cut $K$ that splits $(Y,Z)$ we have that $K$ refines $K_i$ for some $K_i\in C$.  %

The next lemma not only establishes an important property of splitting cores, but its proof also includes an approach \SK{(summarized in  Figure \ref{fig:split-branching})} to obtain such a splitting core. 
This is then used in \cref{rule:overlap} to branch with respect to the cuts in the splitting core. \SK{See Figure \ref{fig:split-def} for more intuition on the concept of a splitting core.}

\begin{lemma}\label{lem:cutweights}
Let $T$ be a binary phylogenetic tree and $(Y,Z)$ a non-trivial bipartition of the labels of $T$.  Then there exists a splitting core $C$ for $(Y,Z)$ such that $\sum_{K\in C} 2^{-|K|} \leq \frac{1}{2}$.
\end{lemma}
\begin{proof}
Induction on the number of leaves of $T$, strengthening the inductive hypothesis to only assume maximum degree 3 rather than that all internal nodes have degree exactly 3.  The base case (one leaf) is vacuous since there is no non-trivial bipartition on one element.

If $T[Y]$ is pendant in $T$, say detached by deleting edge $e$, then $\{\{e\}\}$ is a splitting core satisfying the inequality.  Otherwise, there is an internal node $v$ of $T$ of degree 3 joined by an edge $e_1$ to a pendant subtree whose leaves are all drawn from $Y' \subseteq Y$, by an edge $e_2$ to a pendant subtree whose leaves are all drawn from $Z' \subseteq Z$, and by an edge $e_3$ to a pendant subtree containing a mixture of $Y$ and $Z$ leaves.  Write $T_1$ and $T_2$ for the component of $v$ in $T\setminus \{e_1\}$ and $T\setminus \{e_2\}$ respectively; note that each of $T_1$ and $T_2$ contains both $Y$- and $Z$-nodes.  Let $C_1$ and $C_2$ be splitting cores for $(Y,Z)\cap T_1$ and $(Y,Z)\cap T_2$ satisfying the inequality (by the inductive hypothesis).  Then $C=\{K\cup\{e_1\} | K\in C_1\} \cup \{K\cup \{e_2\} | K\in C_2\}$ is a splitting core for $(Y,Z)$.  Indeed, let $K$ be a cut that splits $(Y,Z)$.  Then $K$ refines $\{e_1\}$ or $\{e_2\}$; say $\{e_1\}$ w.l.o.g.   Then
$K \cap T_1$ splits $(Y,Z)\cap T_1$, so
$K \cap T_1$ refines $K'$ for some $K'\in C_1$, 
and $K$ refines $K'\cup \{e_1\}\in C$, as required.  
Now to complete the proof note that by the inductive hypothesis \[\sum_{K\in C} 2^{-|K|} \leq \sum_{K\in C_1} 2^{-(|K|+1)} + \sum_{K\in C_2} 2^{-(|K|+1)} \leq \frac{1}{4} + \frac{1}{4} = \frac{1}{2}. \qedhere\]
\end{proof}

\begin{branchrule}[SPLIT]\label{rule:overlap} 
Suppose that components $T'_1,T'_2 \in F'$ have overlapping embeddings
in $T$, \SK{i.e. $T[ \ell(T'_1) ]$ and $T[ \ell(T'_2)]$ are
not disjoint.}
We write $T_i$, $i \in \{1,2\}$ as shorthand for these embeddings, i.e. $T_1 = T[ \ell(T'_1) ]$  and $T_2 = T[ \ell(T'_2)]$. Select a common edge $e \in T_1\cap T_2$.  First branch on whether $e$ will be assigned to embedding $T_1$ or embedding $T_2$.
For the branch in which $e$ is assigned to $T_{3-i}$, let $(Y,Z)$ be the bipartition induced on $\ell(T'_i)$ by associating labels which are on the same side of $e$ in $T_i$.  Let $C$ be a splitting core as in the statement of Lemma \ref{lem:cutweights}.
Now branch on the cuts of $C$.
\end{branchrule}

As stated earlier the idea behind this branching rule, which generalizes the DIFFERENT COMPONENTS branching rule\footnote{
To see this, if $\{a,b\}$ is a cherry in $T$, and $a$ and $b$ belong to different non-singleton components in $F'$, then the embeddings of these two non-singleton components in $T$ necessarily overlap on the parent of $a$ and $b$ in $T$. If $a$ or $b$ is a singleton in $F'$, then it would already have been tidied up.
}, 
is that in any agreement forest at most one component can use edge $e$ in $T$. For a component (say, $T'_1$) that does not have edge $e$ available, no path can survive in $T'_1$ from a taxon in $Y$ to a taxon in $Z$. 
Hence, any agreement forest reached from this point necessarily separates all the $Y$ taxa into different components than the $Z$ taxa. Crucially, the branching factor of such a separation procedure is comparatively low.
Formally this is because of the inequality in Lemma \ref{lem:cutweights}; intuitively, once we have decided to separate $Y$ from $Z$ in $T'_1$ we can recursively identify pairs of edges, such that at least one has to be cut in any (optimal) agreement forest reached from this point \SK{(see Figure \ref{fig:split-branching})}. 
The fact that we do not make a cut with the initial branch (i.e. when we decide whether $T'_1$ or $T'_2$ claims the edge $e$ in $T$) is compensated for by the fact that the final cut required to separate $Y$ from $Z$, is deterministic (i.e. ``for free'').

\newcommand{\CC}{\mathcal{C}}
\begin{lemma}
Branching rule \ref{rule:overlap} is safe and has recurrence dominated
by $T(k)=2T(k-1)$, assuming $T(k) \in \omega(2^k)$. This yields branching factor 2.
\end{lemma}
\begin{proof}
    For safety, the branch where $e$ is allocated to $T_{3-i}$ represents the case that edge $e$ is \emph{not} used by any of the components that $T_i$ is split into in the agreement forest ultimately reached; clearly this must hold for either $i=1$ or $i=2$ (or both).  Then the final MAF must be a
    cut which splits $T_i$ with respect to the bipartition $(Y,Z)$
    and so it must refine some element of the splitting core $C$.
    For the recurrence, writing $C_1$ and $C_2$ for the splitting cores for the two branches, we have $T(k) = \left( \sum_{K\in C_1} T(k-|K|))\right) + \left( \sum_{K\in C_2} T(k-|K|)\right)$, which by Lemma~\ref{lem:cutweights} is dominated by $T(k)=2T(k-1)$.
\end{proof}

Note that since the overall recurrence is assumed to be in $\omega(2^k)$ any non-trivial splitting core would only improve the recurrence. 
In particular, encountering a splitting core with more than a single cut of size 1 leads to a greater degree of branching at branching factor 2, and hence an improved overall running time.

\subsection{Decomposing into disjoint components}\label{subsec:decompose}

We call a $(T,F')$ instance such that the components of $F'$ induce disjoint subtrees of $T$ a \emph{disjoint instance}. 
It is easy to see that for a disjoint instance a smallest possible agreement forest reachable from $(T,F')$ can be obtained by independently computing smallest possible agreement forests for each of the tree pair instances $T|B, F'|B$, and combining them, where $B$ ranges over all components in $F'$. 
However, if we are seeking a MAF with at most $k$ cuts, it is not clear how these cuts need to be partitioned among the sub-instances. 
It turns out that we can make gains by first testing each sub-instance at some constant depth $t$, and then (for those for which this fails) re-testing using the knowledge thus gained about the number of cuts needed for other components. 
We note that testing at depth $t$ amounts to asking whether the sub-instance has an agreement forest with at most $(t+1)$ components and if so returning the smallest such agreement forest. 
For constant $t$ this can be computed in polynomial time simply by trying all possible ways to cut at most $t$ edges in one of the trees comprising the sub-instance.

\newtheorem*{rec*}{Recursion rule}
\begin{rec*} [Parameter $t$]
    Let $\mathcal{I} = (T,F')$ be a disjoint instance, let $T'_1\ldots T'_q$ be the components of $F'$, and write $T_i$ for the subtree of $T$ induced by the taxa of $T'_i$.  To solve $\mathcal{I}$ to depth $k$, do the following:
    \begin{enumerate}
        \item\label{it:recstep1} For each $i$, solve the instance $(T_i,T'_i)$ to depth $t$, and let the result be $r_i$.  Write $r'_i=t+1$ if $r_i=\qfalse$ and $r'_i=r_i$ otherwise.
        \item\label{it:recstep2} For each $i$, let $s_i=r_i$ if $r_i\neq \qfalse$.  If $r_i=\qfalse$, solve the instance $(T_i,T'_i)$ to depth $k-\sum_{j\neq i} r'_j$, and let the result be $s_i$.
        \item Return $\qfalse$ if any $s_i=\qfalse$ and $\sum_i s_i$ otherwise.
    \end{enumerate}
\end{rec*}

Note that this is clearly sound: once we have established at step \ref{it:recstep1} that all the components except for $i$ require a total of at least $\sum_{j\neq i} r'_j$ cuts to produce a MAF, it suffices at step \ref{it:recstep2} to solve component $i$ to depth $k-\sum_{j\neq i} r'_j$.
In terms of running time, if all components are homeomorphic then we terminate after step \ref{it:recstep1} in polynomial time.  If exactly one component is non-homeomorphic, then Step \ref{it:recstep2} contains a single call also to depth $k$.  However, if at least two components are non-homeomorphic then we see in the next lemma that we obtain a huge gain, equivalent to applying a branching rule with branching factor $2^{1/(t+1)}$.

\begin{lemma} \label{ref:easydecompose}

    For \DM{$k\geq t+1$}, if the disjoint instance $\mathcal{I} = (T,F')$ has at least two non-homeomorphic components (i.e. those where
    $T_i$ is not homeomorphic to $T'_i$) then the cost of solving instance $\mathcal{I}$ with the recursion rule with parameter $t$ has recurrence dominated by $T(k)=2T(k-(t+1))$.
\end{lemma}
\begin{proof}
    Step \ref{it:recstep1} involves solving the sub-instances to constant depth, and hence only polynomial work.  Let $p$ be the number of sub-instances returning $\qfalse$ at step \ref{it:recstep1} and hence requiring to be solved at step \ref{it:recstep2}.  If $p=1$ then there is only one sub-instance to be solved, to depth at most $k-1$ (since at least one other sub-instance was non-homeomorphic and so required at least one cut), for total work at most $T(k-1)$.  If $p \geq 2$ then each subinstance is solved to depth at most $k-(p-1)(t+1)$ \DM{(or return $\qfalse$ if this is negative)} and so we have total work at most $pT(k-(p-1)(t+1))$.  It is easy to check that this is dominated by $p=2$, giving the required result.
    \end{proof}

Note that the Recursion rule and Lemma \ref{ref:easydecompose} can be further optimised to use simultaneous iterative deepening in place of a constant $t$, giving a bound of essentially $T(k-1)$ rather than $2T(k-(t+1))$.  We do not describe this in detail, however, because it does not improve the overall running time of any of our algorithms: for each one the limiting case (as well as most of the mathematical work) is in handling/avoiding the situation where all except one of the components are homeomorphic.  In fact, for the unrooted and rooted algorithms it will suffice to take $t=0$ and $t=1$ respectively.

\section{Unrooted trees}
\label{sec:unrooted}

Before deploying our split-or-decompose technique to uMAF, we start with some insights and some improvements of special cases.
First, note that when Chen's algorithm is applied to a tree-pair $(T,T')$, rather than a tree-forest pair, there is complete symmetry between $T$ and $T'$. In particular, we can apply the second part of Chen's algorithm to cherries from $T'$, not just to cherries from $T$. (This insight holds for Whidden's algorithm on rooted trees, too).

Next, we observe that CHEN $(t=2)$ can actually be further strengthened.

\begin{restatable}{observation}{obsttwo}
\label{obs:t2}
Let $\{a,b\}$ be a cherry in $T$ such that $a$ and $b$ are in the same component $B$ of $F'$ and the path from $a$ to $b$ in $B$ has exactly two incident edges $e_1, e_2$. Then, if an agreement forest $F^{*}$ can be reached from here by making at most $k' \leq k$ cuts, then with $k'$ (or fewer) cuts an agreement forest can be reached in which $\{a,b\}$ are in the same component.
\end{restatable}
\begin{proof}
Suppose an agreement forest $F^{*}$ is reached from this position where $a$ and $b$ end up in different components $B_a$ and $B_b$. Observe that, due to $\{a,b\}$ being a cherry in $T$, at least one of $B_a$ and $B_b$ is a singleton component. If both $B_a$ and $B_b$ are singleton components, then at most one component distinct from $B_a$ and $B_b$, let us call this $B_c$, intersects the path from $a$ to $b$ in $B$. We can cut this component $B_c$ (incurring one extra cut) to free up the path from $a$ to $b$ in $B$ and then join $B_a$ and $B_b$ (saving one cut). Hence, $a$ and $b$ are now together in a component and the number of components has not increased; we are done. Next, assume that exactly one of $B_a$ and $B_b$ is a singleton, so w.l.o.g. $|B_a| \geq 2$ and $|B_b|=1$. If $|B_a|=2$, we can merge $B_a$ and $B_b$ to obtain a valid agreement forest with one fewer component and we are done. If  $|B_a| > 2$ and at least one of the edges $e_1, e_2$ is not used by $B_a$, it is again safe to simply merge $B_a$ and $B_b$. The only issue is if both $e_1$ and $e_2$ are used by $B_a$. In this case, let $\emptyset \subset S_i \subset (B_a \setminus \{a
\})$, $i \in \{1,2\}$, be the subset of $B_a$ that lies beyond edge $e_i$. Note that $B_a = \{a\} \cup S_1 \cup S_2$. We replace the two components $B_a, B_b$ with the two components $S_1, S_2 \cup \{a,b\}$ and we are done.
\end{proof}

\SK{\textbf{High-level intuition.}} By Observation~\ref{obs:t2}, the case $t=2$ simplifies to: cut one of $\{e_1, e_2\}$ off, giving a bound of~$2^k$. After this mild strengthening, the case $t=3$ is the only limiting case of Chen's algorithm\footnote{For $t \geq 4$ the branching factor of CHEN is at most 2.595.}, with recurrence $2T(k-1)+3T(k-2)$ (yielding branching factor 3). So, to improve the overall
running time for uMAF, it is sufficient to target this
case, and in particular the three branches where two cuts are made.
As discussed in Section \ref{sec:splitordecompose}, if the resulting components are not disjoint then we can immediately apply branching rule \ref{rule:overlap}
and replace a $T(k-2)$ with $2T(k-3)$. \SK{The term $2T(k-3)$ is obtained when we substitute $k-2$ for $k$ in the (worst-case) recurrence for branching rule 
\ref{rule:overlap}, i.e. $2T((k-2)-1)$.}  If the components are disjoint but at least two of them are non-homeomorphic then we also win by Lemma \ref{ref:easydecompose}: \SK{we again obtain a strengthening from $T(k-2)$ to $2T((k-2)-1) = 2T(k-3)$}. \SK{We will in due course argue that at least two of the three $T(k-2)$ terms can be strengthened this way. Then, instead of an overall running time of $2T(k-1)+3T(k-2)$ we get $2T(k-1) + T(k-2) + 4T(k-3)$ which solves to the required branching factor of $\magicunrootednumber$.}
\SK{To make the above machinery work we need to deal with the remaining case when the components resulting from a $T(k-2)$ invocation are disjoint in $T$ and two of them are homeomorphic.} We will handle this by introducing new branching rules with low branching factors which are applicable in certain special cases, and then show that if none of these are applicable then when we apply CHEN $(t=3)$ the bad case cannot happen. \SK{Ultimately this will culminate in \textsc{Algorithm uMAF} shown later in Figure \ref{fig:umaf}}.

The first special-case rule is described by the following lemma.

\begin{restatable}{lemma}{correcttwosingletons}
\label{lem:twosingletons}
Let $T$ and $T'$ be two unrooted trees
on $X$ without common cherries. Assume that CHEN does not apply with $t=2$ or $t\geq 4$. Let $\{a,b\}$ be a cherry in $T$ that allows a $t=3$ application of CHEN and let $\emptyset \subset X_1, X_2, X_3 \subset (X\setminus\{a,b\})$ be the three subsets of taxa
that are pendant on the path from $a$ to $b$ in $T'$ (the labelling of these three subsets is arbitrary). If two or more of the $X_i$ have cardinality 1, then
the branching rule CHEN $(t=3)$ can be slightly adapted to yield recursion $T(k) = 2T(k-1) + 2T(k-2)$ which has a branching factor of at most $\approx 2.73$. Specifically: let $X_1 =\{c\}$ and $X_3 = \{d\}$ and without loss of generality (because the labels of $a$ and $b$ are symmetric, and because of the arbitrary labelling of the $X_i$) assume that $c$ is sibling to $a$ in $T'$. Then the branch of CHEN $(t=3)$ where $c$ and $X_2$ are cut off,
does not need to be considered.
\end{restatable}
\begin{proof}
Observe that $\{a,c\}$ is a cherry in $T'$. Therefore, to prevent triggering CHEN $t=2$ or $t\geq 4$, $c$ must be exactly four edges away from $a$ in $T$. Consider an agreement forest reached from this point that contains the component $\{a,b,d\}$. In this forest, $c$ is definitely a singleton. It can be verified that, due to $c$ being exactly four edges away from $a$ in $T$, replacing the components $\{a,b,d\}$, $\{c\}$ with the components $\{a,b,c\}$, $\{d\}$ also yields a valid agreement forest. Hence, a $T(k-2)$ branch of CHEN $(t=3)$ can be dropped---specifically, the one that cuts off $c$ and $X_2$---and we obtain the required recursion.
\end{proof}

The improvement of CHEN $t=3$ by Lemma~\ref{lem:twosingletons} is denoted CHEN $(t=3, \geq 2 \text{ singletons})$.
Crucially, incorporating this branching rule allows us to focus our attention on applications of CHEN $t=3$ when \emph{at most one} of $X_1, X_2, X_3$ is a singleton. 
The next special-case  branching rule detects and handles the situation that the two trees contain pendant components which are homeomorphic once cut off.

\begin{branchrule}\label{rule:subtree}
    Let $T$ and $T'$ be unrooted trees with no common cherries. Suppose that there exists
    $X' \subset \ell(T')$ such that,
    \begin{enumerate}
    \item $T'|X' = T|X'$ 
    \item $T'[X']$ is pendant in $T'$, and $T[X']$ is pendant in $T$
    \item $|X'| > 1$.
    \end{enumerate}
That is, $X'$ induces a non-trivial common pendant subtree, ignoring rooting, in $T'$ and $T$.
    Then we cut off the pendant subtree $T'[X']$ from $T'$, cut off the pendant subtree $T[X']$ from $T$, and delete the two pendant subtrees we have cut off (since they are homeomorphic).
\end{branchrule}

\begin{restatable}{lemma}{correctsubtree}
    
    Branching rule \ref{rule:subtree} is safe and has recurrence $T(k)=T(k-1)$.
\end{restatable}
\begin{proof}
    The recurrence is trivial, since we deterministically make one cut in $T'$. For safety, observe that
    since all common cherries have already been collapsed, and $|X'|>1$, the two pendant subtrees induced by $X'$ are \emph{not} common including rooting; informally, they are attached to $T'$ and $T$ at different locations. Hence, any agreement forest obtained from this point cannot have a component that is a strict superset of $X'$. Consider, then, a smallest possible agreement forest obtained from this point. 
    Suppose this agreement forest contains a component $B_1$ containing at least one taxon from $X'$ and at least one taxon from outside $X'$. Then there is at least one component $B_2$ that is a strict subset of $X'$. Replacing $B_1$, $B_2$ (and any other components that are strict subsets of $X'$) by $X'$, $(B_1\setminus X')$ yields an agreement forest of the same or smaller size, in which $X'$ is a component, from which safety follows.
\end{proof}

The next lemma shows that the non-applicability of CHEN for certain values of $t$, implies that some of the branches created by CHEN $t=3$ induce forests with overlapping components. This insight will play a crucial role later in the proof of Lemma \ref{ref:monster2}.

\begin{restatable}{lemma}{correctmonster}
\label{lem:monster}
Let $T$ and $T'$ be two unrooted trees on $X$ with no common cherries. Assume that none of CHEN $(t=2)$, $(t\geq 4)$ or $(t=3, \geq 2 \text{ singletons})$ apply. Let $\{a,b\}$ be a cherry in $T$ that allows a $t=3$ application of CHEN and let $\emptyset \subset X_1, X_2, X_3 \subset (X\setminus\{a,b\})$ be the three subsets of taxa
that are pendant on the path from $a$ to $b$ in $T'$ (the labelling of these three subsets is arbitrary). Due to non-applicability of CHEN $(t=3, \geq 2 \text{ singletons})$, at least two of $X_1, X_2, X_3$ will have cardinality 2 or larger. Without loss of generality let $X_1$ be a set with cardinality at least 2.
Suppose $T[ X_1 \cup \{a,b\}]$, $T[ X_2 ]$ and $T[ X_3]$ are mutually disjoint in $T$, and that $T|(X_1 \cup \{a,b\}) = T'|(X_1 \cup \{a,b\})$.  Then: \begin{itemize}
    \item $T[X_2]$ and $T[X_3]$ are both pendant in $T$, but $T[X_1 \cup \{a,b\}]$ is not,
    \item $T[ X_1 ]$, $T[ X_2 \cup \{a,b\} ]$ and $T[ X_3]$  are \underline{not} all mutually disjoint, and
    \item $T[ X_1 ]$, $T[ X_2 ]$ and $T[ X_3 \cup \{a,b\} ]$ are \underline{not} all mutually disjoint.
\end{itemize} 
In particular, at least two of the $T(k-2)$ branches of CHEN $(t=3)$ create forests of $T'$ that induce overlapping components in $T$.
\end{restatable}

\begin{figure}[t]
\centering
\includegraphics{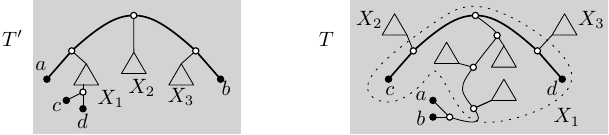}
\caption{The situation described in the proof of Lemma \ref{lem:monster}. Here $X_1$ is the leftmost of the three subtrees between $a$ and $b$ in $T'$, but this is not necessary. In $T$, $T[X_1]$ is shown in the dotted circle.}
\label{fig:overlapUnrooted}
\end{figure}

\begin{proof}
See Figure \ref{fig:overlapUnrooted}. Let $\{c,d\}$ be a cherry in $X_1$ in $T'$. Consider the embedding $T[ X_1 \cup \{a,b\}]$. 
The path from $c$ to $d$ on this embedding has at least three edges, because otherwise $\{c,d\}$ would be a common cherry. Due to the assumption that
$T|(X_1 \cup \{a,b\}) = T'|(X_1 \cup \{a,b\})$ and the mutual disjointness assumption,
the only way that this path can have at least three edges, is if at least one of $T[X_2]$ and $T[X_3]$ is pendant on this path. Suppose only one of them is pendant on the path. Then the path would have length exactly three, so CHEN $(t=2)$ would apply to cherry $\{c,d\}$ - recall that we are free to apply CHEN to $T$ or $T'$ - so we obtain a contradiction.
So both $T[X_2]$ and $T[X_3]$ are pendant on the path. Recalling that $c$ and $d$ are both in $X_1$, and that the path from $c$ to $d$ in $T$ intersects both the path from $\{a,b\}$ to $T[X_2]$, and the path from $\{a,b\}$ to $T[X_3]$, we have that $T[X_1]$ overlaps with $T[ X_2 \cup \{a,b\} ]$, and that $T[X_1]$ overlaps with $T[ X_3 \cup \{a,b\} ]$. 
\end{proof}

Next, we present another new branching rule which is primarily designed to deal with a tricky case we will later encounter in the proof of Lemma \ref{ref:monster2}.

\begin{branchrule}\label{rule:noncorrcomps}
Let $T$ and $T'$ be two unrooted trees on $X$ without common cherries. Assume that CHEN $(t=2)$ (i.e. the stronger version) does not apply. Suppose $X$ can be partitioned into three blocks $A, B, C$ 
satisfying the following
 (as in Figure \ref{fig:unrootedABC}):
\begin{enumerate}
\item In $T$, $T[A]$, $T[B]$ and $T[C]$ are mutually disjoint, and so are $T'[A], T'[B]$ and $T'[C]$ in $T'$.  
\item $T[B]$ is not pendant in $T$, hence $T[A]$ and $T[C]$ are.
\item $T'[A]$ is not pendant in $T'$, hence $T'[B]$ and $T'[C]$ are.
\item $T|B = T'|B$.
\end{enumerate}
Then branch on (1) cutting the edge between $T[A]$ and $T[B]$ in $T$, or (2) the edge between $T[B]$ and $T[C]$.

\end{branchrule}

\begin{figure}[h]
\centering
\includegraphics{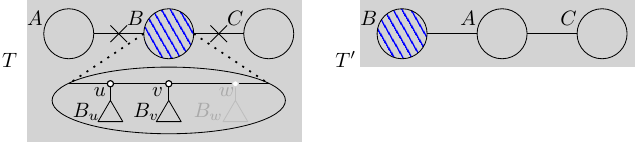}
\caption{Branching rule \ref{rule:noncorrcomps} states that if the depicted situation is encountered, and the strengthened version of CHEN $(t=2)$ does not apply, then at least one of the two crossed edges can be cut in $T$. Here $B$ is shown shaded to emphasize
that $T|B=T'|B$.}
\label{fig:unrootedABC}
\end{figure}

\begin{restatable}{lemma}{unrootedBRlemma}
\label{lem:ubrcorrect}
Branching rule \ref{rule:noncorrcomps} is safe and has recursion $T(k)=2T(k-1)$, yielding a branching factor of $2$.
\end{restatable}
\begin{proof}
The branching factor is at most 2 because we delete one of two edges.

We now prove safeness. Due to non-pendancy, $|A|, |B| \geq 2$. We cannot make any cardinality assumptions on $|C|$, but we do not need them. Observe that in $T$, any path $P$ from a taxon in $A$ to a taxon in $C$ must pass through at least one edge $e=\{u,v\}$ of $T[B]$; this is due to the non-pendancy of $T[B]$ in $T$. The edge leaving $u$ that does not lie on $P$, must therefore feed into a pendant subtree of $T$ containing taxa $\emptyset \subset B_u \subset B$. Define $B_v$ similarly. If $P$ intersects with two or more edges of $T[B]$, then there exists a subpath of $P$, $u-v-w$ where $u,v,w$ all lie on $T[B]$. In this case we define $B_w$ similarly to $B_u, B_v$. However, $B_w$ will not always exist (in which case $B=B_u \cup B_v$).
See Figure \ref{fig:unrootedABC}.

Now, consider any smallest agreement forest $F$ reached from this point. If the edge between $T[A]$ and $T[B]$ was cut to reach this forest, or the edge between $T[B]$ and $T[C]$, we are in branch (1) or (2) and are done, so assume that neither holds. This means that some component $U$ in the forest must contain a taxon from $A$ and a taxon from $B \cup C$, and some component $U'$ must contain a taxon from $C$ and a taxon from $A \cup B$; possibly $U=U'$. We distinguish several cases. In each case we show that an agreement forest exists, no larger than $F$, in which at least one of branch (1), (2) holds, or we derive a contradiction.

Suppose $U$ intersects $A$ and $C$, but not $B$. (In this case $U=U'$). Then $B_u$ and $B_v$ must have been cut off to reach $F$. We make a single cut to $U$ to separate its $A$ part from its $C$ part, and then compensate for this extra cut by introducing $B$ as a single component i.e. merging the two or
more components that are subsets of $B$, into a single component $B$. There are at least two such
components, because $B_u$ and $B_v$ had originally been cut off, so this saves at least one cut. 
Note that we can introduce $B$ as a component because by assumption $T|B = T'|B$. Both branch (1)  and branch (2) now apply, and we are done.

If $U$ intersects $A$, $B$ and $C$ (which again means that $U=U'$) then $U$ cannot intersect with more than one of $B_u$ and $B_v$ (and if it exists, $B_w$), because then $T[U] \neq T'[U]$. If $B_w$ exists, then due to the fact that $U$ cannot intersect with more than one of $B_u, B_v, B_w$, we can cut off the $A$ and $C$ parts of $U$, which costs two cuts, and then introduce $B$ as a single component, saving two cuts, and we are done, as again branch (1) and branch (2) both now apply. Therefore, assume $B_w$ does not exist. 
 Observe that it is not possible for both $B_u$ and $B_v$ to each contain two or more taxa. This is because, if they did, they would each contain a cherry, and at least one of these  would also be a cherry in $T'$ (because $T'[B]$ is pendant), contradicting the assumption that all common cherries have been reduced. Let $x \in B_u$ and $y \in B_v$. If $|B_u| = |B_v|=1$, then $|B|=2$ and $\{x,y\}$ is a cherry in $T'$, because $T'[B]$ is pendant. But $x$ and $y$ are separated by exactly three edges in $T$, so CHEN $(t=2)$ would have held; contradiction. Hence, assume without loss of generality that $|B_u| \geq 2$. Then $B_u$ must, in $T$, contain some cherry $\{x,z\}$. If $\{x,z\}$ is a cherry in $T'$, we obtain a contradiction to the assumption that the cherry reduction does not apply. Given that $T|B = T'|B$, the only way for $\{x,z\}$ not to be a cherry in $T'$, is if it is ``rooted'' on the edge feeding into $x$, or the edge feeding into $z$. Either way, CHEN $(t=2)$ can then be applied to $\{x,z\}$, because $x$ and $z$ have exactly three edges between them in $T'$, and we obtain a contradiction.

The only remaining case is that $U$ and $U'$ are distinct, i.e. $U$ intersects with $A$ and $B$ but not $C$, and $U'$ intersects with $B$ and $C$ but not $A$. However, observe that then $T'[U]$ and $T'[U']$ are not disjoint (because $B$ lies ``inbetween'' $A$ and $C$ in $T$, but in $T'$, $A$ lies inbetween $B$ and $C$). So this case cannot occur.
\end{proof}

\begin{figure}
  \caption{\textsc{Algorithm uMAF.} 
\begin{enumerate}
\item Apply tidying-up operations %
\item\label{it:uprerules} Apply the first among the following branching rules to be applicable:
\begin{enumerate}
    \item Branching rule \ref{rule:subtree} (1)
    \item Branching rule \ref{rule:noncorrcomps} (2)
    \item CHEN $(t=2)$ (strong version with branching factor 2, see Observation \ref{obs:t2})
    \item CHEN $(t \geq 4)$ (2.595)
    \item CHEN $(t = 3, \geq 2 \text{ singletons})$ (2.73)
    \item CHEN $(t=3)$ (3)
    \end{enumerate}
    \item Apply Branching rule \ref{rule:overlap} \SK{(SPLIT)} to exhaustion (2)
    \item Apply the recursion rule \SK{(i.e. decomposition)} with $t=0$
\end{enumerate}
}
\label{fig:umaf}
\end{figure}

We are now ready to describe the overall algorithm for unrooted trees: see \textsc{Algorithm uMAF} in Figure \ref{fig:umaf}. Values in brackets denote the branching factor of each branching rule.

Observe that all of these branching rules have branching factors below our target of $\magicunrootednumber$ apart from CHEN $(t=3)$, which has branching factor 3.  Thus, it remains to show that if CHEN $(t=3)$ is reached then either subsequent applications of branching rule \ref{rule:overlap}, or breaking into many non-homeomorphic components and applying Lemma \ref{ref:easydecompose}, will result in an overall branching factor at most~$\magicunrootednumber$.

\begin{lemma}
\label{ref:monster2}
Let $T$ and $T'$ be two unrooted trees on $X$ with no common cherries. Assume that none of branching rule \ref{rule:subtree}, branching rule \ref{rule:noncorrcomps}, CHEN $t=2$, $t\geq 4$ or $(t=3, \geq 2 \text{ singletons})$ apply.  Then applying CHEN $(t=3)$ followed by branching rule \ref{rule:overlap} to exhaustion, followed by applying the recursion rule with $t=0$, has overall branching factor at most~$\magicunrootednumber$.

\end{lemma}
\begin{proof}
Let $\{a,b\}$ be a cherry in $T$ that allows a $t=3$ application of CHEN and let $\emptyset \subset X_1, X_2, X_3 \subset (X\setminus\{a,b\})$ be the three subsets of taxa that are pendant on the path from $a$ to $b$ in $T'$ (the labelling of these three subsets is arbitrary). Due to non-applicability of CHEN $(t=3, \geq 2 \text{ singletons})$, at least two of $X_1, X_2, X_3$ will have cardinality 2 or larger. Moreover, due to non-applicability of branching rule \ref{rule:subtree}, no $X_i$ with $|X_i|\geq 2$ can have the property that $T[X_i]$ is pendant in $T$, $T'[X_i]$ is pendant in $T'$, and $T|X_i = T'|X_i$. 

We analyse the running time of the CHEN $t=3$ invocation. As usual, we incur a $2T(k-1)$ term, stemming from cutting off $a$, or cutting off $b$. We do not attempt to optimize this. Rather, consider the three $T(k-2)$ instances that CHEN $(t=3)$ usually creates. Suppose some $X_i$, with $|X_i|\geq 2$, say $X_1$, induces the situation described in the statement of Lemma \ref{lem:monster}. Then by Lemma \ref{lem:monster} the two $T(k-2)$ branches corresponding to $X_2$ joining with $\{a,b\}$, and $X_3$ joining with $\{a,b\}$, are non-disjoint, so branching rule \ref{rule:overlap} can be applied to each of them. Recalling that branching rule \ref{rule:overlap} has a recursion $2T(k-1)$, this means that each of the $T(k-2)$ branches needs time at most $2T((k-2)-1) = 2T(k-3)$. Hence, the total time for both these branches is at most $4T(k-3)$. We simply spend time $T(k-2)$ on the instance in which $X_1$ joins with $\{a,b\}$, and we obtain a recurrence $2T(k-1) + T(k-2) + 4T(k-3)$, which solves to the required bound of $\magicunrootednumber$.

Assume, then, that none of the sets $X_i$ with cardinality 2 or higher, induce the situation described in Lemma \ref{lem:monster}. Without loss of generality let $X_1$ and $X_2$ be sets with cardinality 2 or higher; we allow the possibility that $X_3$ has cardinality 1, and return to this later.

The only way for (say) the $X_1 \cup \{a,b\}$ subinstance not to trigger Lemma~\ref{lem:monster}, is if (a) the $T[X_1 \cup \{a,b\}], T[X_2], T[X_3]$ subinstances overlap somewhere in $T$ (i.e. are not disjoint), or (b) they \emph{are} disjoint, but $T|(X_1 \cup \{a,b\}) \neq T'|(X_1 \cup \{a,b\})$. In situation (a), branching rule \ref{rule:overlap} applies, meaning that time at most $2T((k-2)-1) = 2T(k-3)$ is required for this branch. The latter case, situation (b), is more subtle. We know that $T|(X_1 \cup \{a,b\}) \neq T'|(X_1 \cup \{a,b\})$, so if we can additionally show that at least one of $T|X_2 \neq T'|X_2$, $T|X_3 \neq T'|X_3$ is true, the decomposition procedure from Lemma~\ref{ref:easydecompose} will kick in, meaning that  time at most $2T((k-2)-1)=2T(k-3)$ is needed for this branch. Noting that at least one of $T[X_2], T[X_3]$ is pendant in $T$ (due to disjointness), and that branching rule \ref{rule:subtree} was not applicable, the conclusion that at least one of $T|X_2 \neq T'|X_2$, $T|X_3 \neq T'|X_3$ holds is immediate if $|X_3| \geq 2$ (because then all three of $X_1, X_2, X_3$ contain two or more leaves). Hence, we can assume $|X_3|=1$. If $T|X_2 \neq T'|X_2$ we also have our desired two non-homeomorphic subinstances, and are done, so assume that $T|X_2 = T'|X_2$. Moreover, to avoid triggering branching rule \ref{rule:subtree} on $X_2$, $T[X_2]$ is \emph{not} pendant in $T$. This means that $T[X_1 \cup \{a,b\}]$ and (vacuously) $T[X_3]$ are pendant in $T$, that $T[X_1 \cup \{a,b\}]$ is incident to $T[X_2]$, and $T[X_3]$ is incident to $T[X_2]$. By construction, $T'[X_1 \cup \{a,b\}]$ is not pendant in $T'$, but $T'[X_2]$ and $T'[X_3]$ are. This means that branching rule \ref{rule:noncorrcomps} could have been applied, by setting in the statement of that branching rule $A, B, C$ to $X_1 \cup \{a,b\}, X_2, X_3$ respectively, contradiction. Hence, $T|X_2 \neq T'|X_2$ and we indeed have our second non-homeomorphic subinstance. Putting it all together: this means that situation (b) immediately triggers the decomposition procedure inherent in Lemma~\ref{ref:easydecompose}, which has a branching factor of 2, yielding just like situation (a) a branching factor of $2T(k-3)$. 
Thus, this branch takes time at most $2T(k-3)$.

The analysis for $X_2 \cup \{a,b\}$, is entirely symmetrical to that of $X_1 \cup \{a,b\}$.

Hence, the branches corresponding to $X_1 \cup \{a,b\}$ and $X_2 \cup \{a,b\}$ each need at most time $2T(k-3)$, and thus at most $4T(k-3)$ in total. If $|X_3| \geq 2$, then the $X_3 \cup \{a,b\}$ branch will also require time at most $2T(k-3)$ (because the same analysis as above holds here too). If $|X_3|=1$, we can pessimistically use running time at most $T(k-2)$ for the $X_3 \cup \{a,b\}$ branch. Noting that $T(k-2)$ dominates over $2T(k-3)$, we obtain a total recurrence of at most $2T(k-1) + T(k-2) + 4T(k-3)$, which again solves to the required $\magicunrootednumber$.
\end{proof}

\begin{theorem}
\label{thm:stevenumaf}
\textsc{Algorithm uMAF} solves the problem uMAF in time $O^{*}({\magicunrootednumber}^k)$.
\end{theorem}

Finally, before moving on to rooted trees, we note that our split-or-decompose technique can be used to improve the existing $O^{*}(2.49^k)$ algorithm for two unrooted \emph{caterpillars}.
\begin{restatable}{lemma}{catlemma}
\label{lem:ruben}
By strengthening the algorithm from \cite{kelk2024agreement} with branching rule \ref{rule:overlap} we obtain an algorithm with running time $O^{*}(2.4634^k)$ for computing uMAF on two unrooted caterpillars.
\end{restatable}
\begin{proof}
 The limiting case for the existing algorithm from \cite{kelk2024agreement} is shown in Figure \ref{fig:ruben}.

\begin{figure}[h]
\centering
\includegraphics[scale=1]{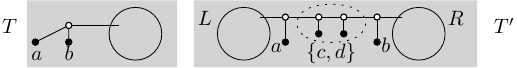}
\caption{The limiting case for the algorithm from \cite{kelk2024agreement} is when $\{a,b\}$ is a cherry in $T$, there are exactly two taxa, $c$ and $d$, between $a$ and $b$ in $T'$, and $L$ and $R$ are both non-empty.}
\label{fig:ruben}
\end{figure}
First, recall that the standard Chen $(t=4)$ algorithm gives recursion $2T(k-1) + 4T(k-3)$. This can be improved by observing that one of the two $T(k-3)$ branches, which create component $\{a,b,c\}$ and $\{a,b,d\}$ respectively, can be dropped. Specifically, we can restrict ourselves to the branch $\{a,b,x\}$ where %
$x\in\{c,d\}$ and $x$ is nearest to the $\{a,b\}$ cherry in $T$ (as it is least disruptive to the rest of $T$). This improves the recursion to $2T(k-1) + 3T(k-3)$ which gives a branching factor of at most $2.49$. Note that if $L$ is a singleton set, or $R$ is a singleton set, we can drop at least one other branch, again prioritizing the taxon that is closest to the $\{a,b\}$ cherry in $T$. Hence, if at least one of $L$ and $R$ is a singleton, we obtain recursion $2T(k-1) + 2T(k-3)$ which gives a branching factor of 2.36.
So suppose that $L$ and $R$ are both non-singleton. If any of the three $T(k-3)$ branches trigger overlap, then by branching rule \ref{rule:overlap} the $T(k-3)$ term improves to $2T(k-4)$ giving an overall running time of $2T(k-1)+2T(k-3) + 2T(k-4)$ which yields $2.4634$. So assume none of the three $T(k-3)$ branches induce overlap: all of them create disjoint instances. One such disjoint instance creates components $\{a,b,x\}$, $L$ and $R$ (where $x$ is as defined before), so $T[L]$ and $T[R]$ are disjoint and each spans at least 2 taxa. Suppose that $L$ is closer
to the $\{a,b\}$ cherry in $T$ than $R$. But then the branch which creates components $\{a,b\} \cup R$ and $L$ is overlapping: contradiction. A symmetrical argument holds if $R$ is closer to the $\{a,b\}$ cherry than $L$. %
\end{proof}

\section{Rooted trees}
\label{sec:rooted}

\noindent
We now switch to rooted trees. We show how rMAF can be solved in time $O^{*}(\bestrootednumber^{k})$. This is a very slight improvement on the limiting recursion from \cite{chen2015faster}, which is $T(k) = T(k-2)+10T(k-3)+2T(k-5)+2T(k-6)+1$. That has a branching factor of $2.3431$. An advantage of our approach is that significantly less deep case analysis is required.

A crucial difference with the unrooted case is that two rooted trees on 3 taxa can  be non-homeomorphic. For example, when $X=\{a,b,c\}$, $T$ has a unique cherry $\{a,b\}$ and $T'$ has a unique cherry $\{b,c\}$. This leads to a number of subtle differences with the unrooted case. Nevertheless, at a high-level we use the same approach as for unrooted trees:
all the results from Section \ref{sec:splitordecompose} hold for both types of trees. \\
\\
\SK{\textbf{High-level intuition.}} Recall that the limiting case of Whidden's algorithm is $t=2$, with recursion $2T(k-1) + T(k-2)$. \SK{Our overall goal is to ensure that the $T(k-2)$ invocation triggers a situation with recurrence $T(k-1) + 2T(k-2)$. By substituting $k-2$ for $k$, this becomes $T((k-2)-1) + 2T((k-2)-2)$ which is $T(k-3) + 2T(k-4)$. Combining with the $2T(k-1)$ term we obtain an overall recurrence
of $2T(k-1) + T(k-3) + 2T(k-4)$: this solves to the desired branching factor of $\bestrootednumber$.}

Analogous to our approach for uMAF, we also need
a number of branching rules which are designed to catch certain special cases -- \SK{we will start by describing these}. \SK{Ultimately, once everything is put together, we will end up with the algorithm described in Figure \ref{alg:mraf}.}

Branching rule \ref{rule:unify}
has a function similar to that of Lemma \ref{lem:twosingletons} for unrooted trees: to ensure that ``not too many'' singleton components are created in the $T(k-2)$ branch of Whidden $(t=2)$.

\begin{figure}[bt]
\centering
\includegraphics{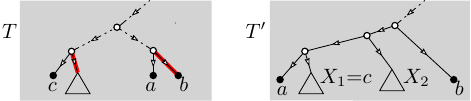}
\caption{By Branching rule \ref{rule:unify}, when this situation occurs, and Whidden $(t=1)$ does not apply, then either $X_1$ and $X_2$ are cut off in $T'$, $a$ is cut off, or in $T$ at least two outgoing arcs (shown in bold red) on the path from $a$ to $c$ are cut. Note that in $T'$, $X_2$ might alternatively be sibling to $b$.
}
\label{fig:firstrooted}
\end{figure}

\begin{branchrule}\label{rule:unify}
Let $T$ and $T'$ be two rooted trees on $X$ without common cherries, and assume that Whidden $(t=1)$ does not apply. Let $\{a,b\}$ be a cherry in $T$ and let $\emptyset \subset X_1, X_2 \subset (X\setminus\{a,b\})$ be the two subsets of taxa
that are pendant on the path from $a$ to $b$ in $T'$, where $X_1$ is closer to $a$, than $X_2$ is to $a$. Suppose $X_1 = \{c\}$ and $\{a,c\}$ is a cherry in $T'$. (Note that if $X_2$ is a singleton, and forms a cherry with $b$, then we can swap the labels of $a$ and $b$ to arrive in the same situation). Then branch as follows:
    \begin{enumerate}
        \item Cut off $X_1$ and $X_2$.
        \item Cut off $a$.
        \item In $T$, cut off the at least two arcs whose tails lie on the path from $a$ to $c$.
    \end{enumerate}
\end{branchrule}

\begin{restatable}{lemma}{correctunifysafe}
\label{lem:unifysafe}
Branching rule $\ref{rule:unify}$ is safe and has a recursion $T(k-1)+2T(k-2)$, which yields a branching factor of at most 2.
\end{restatable}
\begin{proof}
To see that the third branch needs to cut at least two arcs, note that there is an arc outgoing to $b$ whose tail lies on the path from $a$ to $c$ in $T$. If there are no other arcs outgoing from this path then Whidden $(t=1)$ would apply, contradiction. (See also Figure \ref{fig:firstrooted}). For safety, suppose none of the three branches apply. Then $a$ is in a non-singleton component $B$. If $B$ contains $b$, then we are in the first branch. If $B$ does not contain $c$ or any taxa from $X_2$, we are also in the first branch. If $B$ contains taxa from $X_2$, but not $c$, we can cut off $X_2$ and add $b$ (which will necessarily be a singleton) to $B$, irrespective of whether $B$ contains taxa from outside $\{a,b,c\} \cup X_2$: again, we are back in the first branch. The remaining case is that $B$ includes $c$. Then (due to $\{a,c\}$ being a cherry in $T')$ the
two arcs outgoing from the path from $a$ to $c$ in $T$ cannot be in $B$ i.e. they must be cut. Hence, the third branch applies.
\end{proof}

Consider a potential invocation of Whidden $(t=2)$, where $X_1$ and $X_2$ are as defined
in the statement of branching rule \ref{rule:unify}, and let $S=X \setminus (\{a,b\} \cup X_1 \cup X_2)$. (So,
the branch of Whidden $(t=2)$ that cuts off $X_1$ and $X_2$, splits $T'$ into three components with taxa $X_1, X_2$ and $S \cup \{a,b\}$). Assuming branching rule \ref{rule:unify}
does not apply, $X_1$ and $X_2$ are not both singletons. Also, if one of them is a singleton,
it does not form a cherry with $a$ or $b$. 
We are now almost ready to invoke a specially-modified version of Whidden $(t=2)$, but we have to first deal with a very specific boundary situation when both the following conditions hold:
\begin{enumerate}
\item $T[X_1]$, $T[X_2]$ and $T[ S \cup \{a,b\}]$ are mutually disjoint in $T$.
\item At least two of $X_1, X_2, S\cup \{a,b\}$ are homeomorphic.
\end{enumerate}

\newcommand{\hiddenRule}[1]{%
    \refstepcounter{branchrule}%
    \label{#1}%
}

\hiddenRule{rule:twohomeoroot}

The only purpose of our next branching rule, \textbf{Branching rule \ref{rule:twohomeoroot}}, is to catch this situation.  For convenience, we give the definition of this branching rule at the same time as the proof of its correctness; to guide the reader we show the actual branching strategy in \textcolor{red}{red text}.

\begin{lemma}
\label{lem:proofOfBR5}
Let $T$ and $T'$ be two rooted trees on $X$ without common cherries. Suppose Whidden $(t=1)$ does not apply, and neither does branching rule \ref{rule:unify}. Suppose Whidden ($t=2$) does apply. That is: $\{a,b\}$ is a cherry in $T$ and $\emptyset \subset X_1, X_2 \subset (X\setminus\{a,b\})$ are the two subsets of taxa
that are pendant on the path from $a$ to $b$ in $T'$, where $X_1$ is closer to $a$ than $X_2$ is to $a$. Let $S=X \setminus (\{a,b\} \cup X_1 \cup X_2)$. Suppose that $T[X_i], T[X_j]$ and $T[S \cup \{a,b\}]$ are mutually disjoint in $T$, and at least two of $X_i, X_j, S \cup \{a,b\}$ are
homeomorphic.  Then there is a safe branching rule, denoted \textbf{branching rule \ref{rule:twohomeoroot}}, with recursion at most $T(k-1)+3T(k-2)$, which yields a branching factor of at most 2.3028. 
\end{lemma}
\begin{proof}
First, recall that due to the non-applicability of branching rule \ref{rule:unify}, at most one of $X_1$ and $X_2$ is a singleton, and if one of them is a singleton it does not form a cherry with $a$ or $b$.  

Now, if all three of $X_1, X_2, S \cup \{a,b\}$ are homeomorphic, the
$T(k-2)$ branch of Whidden $(t=2)$ would simply return TRUE/0 instantly, because the three sets already constitute a valid agreement forest. \textcolor{red}{Hence in this situation we can simply execute the standard Whidden $(t=2)$ branching, observing that the usual recursion, which is $2T(k-1) + T(k-2)$, collapses to $2T(k-1)$.} This has a branching factor of 2, which is less than 2.3028.

From this point on, we can henceforth assume that \emph{exactly two} of the three sets are homeomorphic. 
Figure \ref{fig:startbigrooted} describes the situation at this point.

\begin{figure}[h]
\centering
\includegraphics{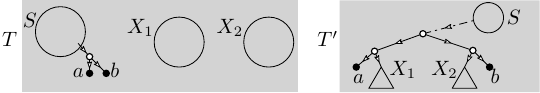}
\caption{The starting point for our technical analysis. Here $X_1$ and $X_2$ are not both singletons; $a$ and $b$ do not form cherries in $T'$ with $X_1$ or $X_2$; $T[X_1]$, $T[X_2]$ and $T[ S \cup \{a,b\}]$ are mutually disjoint in $T$, and exactly two of $X_1, X_2, S \cup \{a,b\}$ induce homeomorphic subtrees in $T$ and $T'$. Here $X_2$ is shown as being sibling to $b$ in $T'$ but it might also be ancestral to $a$ and $X_1$ i.e. ``above'' $X_1$ in $T'$.}
\label{fig:startbigrooted}
\end{figure}
  
 Suppose $X_i$, (say) $X_1$, is homeomorphic
 and $|X_1|\geq 2$. Let $\{p,q\}$ be a cherry in $X_1$; this must exist. To prevent $\{p,q\}$ being a common cherry in $T$, and assuming that Whidden ($t=1$ or $t\geq 3$) does not apply anywhere, we observe crucially that $T[\{p,q\}]$ \emph{must have at least two outgoing arcs in $T$}. These outgoing arcs can only be created by edges outgoing to $T[X_2]$
 and $T[S \cup \{a,b\}]$. This has several consequences: $X_1$ has at most one cherry (so $T|X_1 = T'|X_1$ is in fact a caterpillar) and $X_2$ cannot be homeomorphic non-singleton (as $T'|X_2$ would then also contain a cherry, meaning there are simply too many cherries in $T$ to `hit'). See Figure \ref{fig:pqHitting}.

\begin{figure}[h]
\centering
\includegraphics{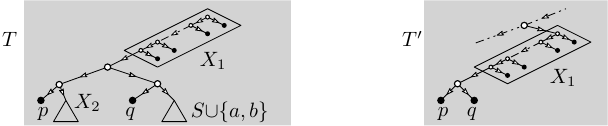}
\caption{If one of $X_1, X_2$ is homeomorphic non-singleton, say $X_1$, then it contains a cherry $\{p,q\}$, and in $T$, $X_2$ and $S \cup \{a,b\}$ will be pendant to the path from $p$ to $q$. $X_1$ will also be a caterpillar, and $X_2$ will not be homeomorphic non-singleton. Note that the location of $X_2$ and $S \cup \{a,b\}$ on the path from $p$ to $q$ can vary.
}
\label{fig:pqHitting}
\end{figure}

Hence, if $X_1$ and $X_2$ each have at least 2 taxa, then at most one of them is homeomorphic. Conversely, if they are both homeomorphic, at least one of them is a singleton. 

We can distinguish the following cases.
\begin{itemize}
\item\textbf{Case 1: At least one of $X_1, X_2$, say $X_j$, is homeomorphic and $|X_j|\geq 3$.} It follows that $X_i$ ($i \neq j$) is either a singleton, or non-homeomorphic, but we do not directly need this information. Similarly, the following argument is agnostic to the (non-)homeomorphic status of $S \cup \{a,b\}$. 
We know that $X_j$ is a non-singleton caterpillar with a unique cherry $\{p,q\}$.  Furthermore, in $T$, $T[X_i]$ and $T[S \cup \{a,b\}]$ must be pendant to the embedding of this cherry (see the remark in the previous paragraph; also,
in Figure \ref{fig:pqHitting}, $X_j = X_1$ and $X_i=X_2$). As a result of this (a) any component of a forest that contains at least one taxon from $X_i \cup S \cup \{a,b\}$, can contain at most one taxon from $X_j$. This is because, if such a component contained two or more taxa from $X_j$, the component would have a cherry comprising two taxa from $X_j$ (due to the pendancy of $T'[X_j])$. However, due to the topology of $T$ - recall that $T|X_j$ is a caterpillar and $T[X_i], T[S \cup \{a,b\}]$ are outgoing from its only cherry - these two taxa cannot form a cherry in the embedding of the component in $T$.
Also, (b) there can be at most one component that contains taxa from both $X_j$ and outside it, due to the pendancy of $T'[X_j]$. \textcolor{red}{We use insights (a) and (b) to now argue that both $T[X_i]$ and $T[S \cup \{a,b\}]$ can definitely be cut off (which is equivalent to cutting off $T'[X_i]$ and $T'[X_j]$).} Suppose this is not the case for some agreement forest.
First, suppose $B$ contains
at least one taxon from $X_i$, and at least one taxon from $S \cup \{a,b\}$ but no taxa from $X_j$. Then at least one of $p, q$ will be a singleton. We cut $B$ into its $X_i$ and $S \cup \{a,b\}$ parts with a single cut, and then compensate by introducing $X_j$ as a single component (which saves at least one cut, due to absorbing singleton $p$ or $q$). On the other hand, suppose that there is a component $B$
that contains at least one taxon from $X_i \cup S \cup \{a,b\}$ and exactly one taxon from $X_j$.  Crucially, all taxa in $(X_j \setminus B)$ will be singleton components\footnote{This means that, again, at least one of $p$ and $q$ will be a singleton. We do not need this fact here, but it is leveraged in the next case.}, since $B$ enters $T'[X_j]$ via its root (due to pendancy of $T'[X_j]$) but $T[X_j]$ via the cherry at the bottom. Given that $|X_j|\geq 3$, there are at least two such singletons. We can then cut $B$ into its $X_i$, $S \cup \{a,b\}$ and $X_j$ parts (requiring at most two cuts), and compensate by introducing $X_j$ as a single component (which absorbs the at least two singletons); done. The recursion is $T(k) = T(k-2)$,
which has a branching factor of 1.

\item\textbf{Case 2: At least one of $X_1, X_2$, say $X_j$, is homeomorphic and $|X_j| = 2$.} The analysis for $|X_j|\geq 3$ from Case 1 holds almost entirely here, except at the very end when we leverage the $|X_j| \geq 3$ assumption. To avoid extensive case analysis, we accommodate this via a comparatively weakened analysis. Our branches are now: \textcolor{red}{(1) Cut off $X_1, X_2$ (which is the situation discussed in the previous paragraph); (2) Cut off $p$ and $a$, (3) Cut off $p$ and $b$, (4) Cut off $q$ and $a$, (5) Cut off $q$ and $b$.} The argument is straightforward: if $a$ and $b$ are together in a component, then via the standard Whidden $(t=2)$ argument, (1) holds. Otherwise, at least one of $a$ and $b$ is a singleton. Now, we know from the previous paragraph that if (1) does not hold, at least one of $p$ and $q$ will be a singleton. The recursion is $T(k)=5T(k-2)$, which has a branching factor of 2.2361.

\item \textbf{Case 3: Exactly one of $X_1, X_2$ is non-homeomorphic, and the other is a singleton.}  Hence, $S \cup \{a,b\}$ is homeomorphic. Now, suppose $T|(S \cup \{a,b\}) = T'|(S \cup \{a,b\})$ contains a cherry $\{p,q\}$ such that $\{p,q\} \cap \{a,b\} = \emptyset$. In this case,    \textcolor{red}{the same branching from Case 2 works.} That is because $X_1$ and $X_2$ must both be pendant to $T[\{p,q\}]$, by the same arguments as given in Case 1. So $X_1, X_2$ are pendant to the cherry $\{p,q\}$ in $T$, and are pendant to the cherry $\{a,b\}$ in $T'$. If $\{p,q\}$ are together in a component, or $\{a,b\}$ are together in a component, then $X_1, X_2$ are both cut off, i.e. branch (1) holds. Otherwise at most one of $p$ and $q$ can be in a non-singleton component, and at most one of $a, b$ can be in a non-singleton component, i.e. at least one of branches (2)-(5) holds. 

We may thus assume that $S \cup \{a,b\}$ has no second cherry, i.e. $T|(S\cup \{a,b\}) = T'|(S\cup \{a,b\})$ is a caterpillar. This situation is shown in Figure \ref{fig:homeomorphicS}.
\begin{figure}[h]
\centering
\includegraphics{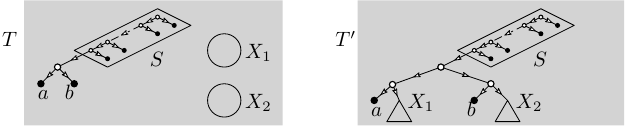}
\caption{Assuming the situation in Figure \ref{fig:pqHitting} does not occur, we have that one of $X_1, X_2$ is non-homeomorphic, and the other is a singleton. Hence, $S \cup \{a,b\}$ is homeomorphic. The most technically challenging situation to deal with is when $S \cup \{a,b\}$ is a caterpillar, as shown here. Different subcases occur depending on how $X_1, X_2$ and $S \cup \{a,b\}$ connect to each other within $T$. These are elaborated on later
in Figures \ref{fig:commonSplit} and \ref{fig:threeblocks}.}
\label{fig:homeomorphicS}
\end{figure}

Now, let $X_j$ be the non-homeomorphic set.
\begin{itemize}

\item \textbf{Case 3.1: $T[X_j]$ has exactly one incident edge, incoming or outgoing.} Observe that in both $T$ and $T'$ the taxa from $X_j$ are completely separated from the taxa in $X \setminus X_j$ (in phylogenetics this is called a ``common split''); this is shown schematically in Figure \ref{fig:commonSplit}.%

\begin{figure}[h]
\centering
\includegraphics{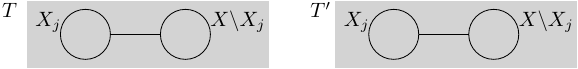}
\caption{Elaborating on Figure \ref{fig:homeomorphicS}, this is a schematic depiction of the special situation when $T[X_j]$ has exactly one incoming or outgoing arc in $T$: $X_j$ (the non-homeomorphic set) is cleanly separated from the rest of the tree, $X \setminus X_j$, in $T$. This separation  occurs by construction in~$T'$.}
\label{fig:commonSplit}
\end{figure}

This can be exploited as follows. \textcolor{red}{We first solve instance $T|X_j, T'|X_j$ with a $T(k-1)$ call, and then solve instance $T|(X_j \cup \rho), T'|(X_j \cup \rho)$ with a second $T(k-1)$ call, where $\rho$ is a new taxon which is attached to $T|X_j$ and $T'|X_j$ to represent the location at which the rest of $T, T'$ (respectively) is attached}. The invocation with $\rho$ allows us to understand whether optimal agreement forests for $T|X_j, T'|X_j$ can potentially grow beyond $X_j$ without sacrificing optimality. In more detail, let $y, y^{\rho}$ be the numerical values returned from the
two invocations (or FALSE if the corresponding invocation failed). If $y$ is FALSE, then an agreement forest for $T|X_j, T'|X_j$ has at least $k+1$ components (i.e. at least $k$ cuts are needed). However, given that $T|(X \setminus X_j)$ and $T'|(X \setminus X_j)$ are not homeomorphic (as $\{a,b\}$ is a cherry in one and not the other), 
an agreement forest
for $T|(X \setminus X_j)$ and $T'|(X \setminus X_j)$ has at least two components. At most one component can span both the $X_j$ and the $X \setminus X_j$ instance, so together at least $(k+1)+2 - 1 =k+2$ components are required, and thus at least $k+1$ cuts, so we can definitely return FALSE. Otherwise: if $y < y^{\rho}$, or $y$ is numeric and $y^{\rho}$ is FALSE, then no optimal agreement forest for $T|X_j, T'|X_j$ can grow beyond $X_j$. Hence, it is safe to assume that in some optimal agreement forest reached from this point, $X_j$ is cut off. Noting that $T|(X \setminus X_j), T'|(X \setminus X_j)$ has an agreement forest with at most two components (due to $X_i$ and $S \cup \{a,b\}$ being disjoint and homeomorphic), we can use brute force to determine whether 0 or 1 cut is required
for this instance, add this to $y$, and return this as our (guaranteed successful) solution. If $y=y^{\rho}$, then we replace $X_j$ with a taxon $\gamma$ in $T, T'$ to obtain $T_\gamma, T'_\gamma$. The instance $T_\gamma, T'_\gamma$ has an agreement
forest with at most 3 components. Let $z$ be the
minimum number of cuts to obtain an agreement forest for $T|(X \setminus X_j), T'|(X \setminus X_j)$, and let $z_\gamma$ be the
minimum number of cuts to obtain an agreement forest for $T_\gamma, T'_\gamma$ (both $z, z_\gamma$ can be computed using brute force,
since $z \leq z_\gamma \leq 2$). If $z < z_\gamma$ we return $y+z$ as our solution (or FALSE if $y+z > k$), but if $z=z_\gamma$ we return $y+z-1$ (or FALSE if $y+z-1 > k$) as our solution. The ``$-1$'' represents the situation when a component from an optimal agreement forest for $T|(X \setminus X_j), T'|(X \setminus X_j)$, and a component from an optimal agreement forest for $T|X_j, T'|X_j$, can be merged into a single component. The recursion here is $T(k)=2T(k-1)$, which has a branching factor of 2.

\item \textbf{Case 3.2: $T[X_j]$ has exactly \emph{two} incident arcs.}
The only situation we need to consider is when $X_j$ (which is not a singleton, due to being non-homeomorphic) is sibling to $a$ in $T'$, $c$ (where $X_i = \{c\}$) is a child of the grandparent of $a$ and $X_j$ in $T'$, and  $b$ is directly incident to the root of $T'[X_j]$ -- see Figure \ref{fig:threeblocks} (right). (All other topologies will trigger branching rule \ref{rule:unify}).  This
again splits into several subcases, contingent on the topology of $T$; these are shown in panels 1, 2, 3 of Figure \ref{fig:threeblocks} (left).

\begin{figure}[h]
\centering
\includegraphics{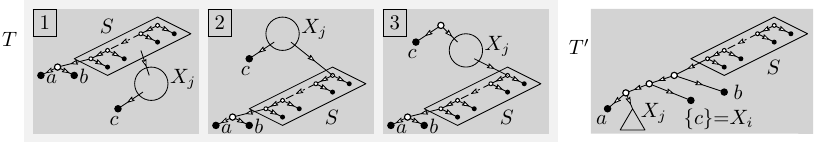}
\caption{Here $T[X_j]$ has exactly \emph{two} incident arcs; the different ways that this is possible are shown in the panels 1, 2, 3 on the left. Assuming other reduction rules do not hold the only topology
we need to consider for $T'$ is shown on the right. (Recall that, due to an earlier assumption, $S \cup \{a,b\}$ is
homeomorphic, and $X_j$ is not).}
\label{fig:threeblocks}
\end{figure}

\begin{itemize}
\item \textbf{Case 3.2.1: In $T$, $c$ is outgoing from $T[X_j]$, and $T[X_j]$ is outgoing from $T[S \cup \{a,b\}]$.} Observe that it is not possible for a component of the forest to  simultaneously include $a, c$ and a taxon from $X_j$. This is because in $T$
$c$ and $X_j$, but not $a$, are together below some cut-edge of $T$, but in $T'$ a different topology is induced (due to $a$ being sibling to $X_j$). We therefore branch as follows: \textcolor{red}{(1) Cut $X_i = \{c\}$ and $X_j$, %
(2) Cut $a$}.
This is safe because if $a$ is not
a singleton component (i.e. branch (2) does not apply), and $a$ is not together with $b$ in a component (i.e. branch (1) does not apply), then $b$ is singleton, and $a$ can be together in a component $B$ with $c$, or together with taxa from $X_j$, or neither -- but not both. Note also that if $a$ is in a component $B$ with taxa from $X_j$ (and therefore $c$ is not in that component), then $c$ is singleton; if $c$ was not singleton, its component $B' \neq B$ would necessarily include a taxon from $S$, but then $T[B]$ and $T[B']$ would intersect on the edge from $T[S \cup \{a,b\}]$ to $T[X_j]$.
Hence, we make (at most) one cut to separate $a$ from $X_j$ or $c$, and then compensate by adding $b$ together with $a$, so we are back in branch (1). The recursion is $T(k) = T(k-1) + T(k-2)$, with a branching factor of 1.618.
\item \textbf{Case 3.2.2: In $T$, $c$ and $T[S \cup \{a,b\}]$ are both outgoing from $T[X_j]$.} Here we branch as follows. \textcolor{red}{(1) Cut off $X_1$ and $X_2$, (2) Cut off $b$, (3) Cut off $a$ and in $T$ also cut off: (i) one arc outgoing from the interior of the path from $lca(c,b)$ to $c$, if there are at least two such arcs or, failing that, (ii) all arcs outgoing from the interior of the path from $lca(c,b)$ to $b$ (excluding the arc to $a$), if they exist or, failing that, (iii) the arc incoming to $lca(c,b)$. In branch 3 at least one of (i)-(iii) occurs because otherwise $T[X_j]$ is pendant.} This gives a recursion $T(k-1) + 3T(k-2)$ (branching factor 2.3028), the worst case being obtained in 3(i) when we have to guess which of the least two arcs outgoing from  the path from $lca(c,b)$ to $c$ need to be cut. We do this by selecting an arbitrary size-2 subset of two such outgoing arcs and guessing which of the two is cut (in addition to $a$); this contributes $2T(k-2)$ to the recursion.

Safety is rather subtle here. Suppose $b$ is not a singleton; then, as usual, assuming branch 1 does not hold, $a$ is singleton. If the component that $b$ is in does not contain either $c$ or taxa from $X_j$ we can separate (if necessary) $c$ from $X_j$ using a single cut and compensate by bringing $a$ into the component that $b$ is in (so we are in branch 1). If $b$ is in a component with exactly one of $c$ and taxa from $X_j$, then we can cut off $c$ or $X_j$ from this component and bring $a$ and $b$ back together (branch 1). So assume that $b$ is together with both $c$ and taxa from $X_j$ in some component $B$. Consider $T[B]$. At most one type 3(i) arc can be used by $T[B]$, and no type 3(ii) arcs. Component $B$ cannot contain any taxa from $S \cup \{a,b\}$ other than $b$, since all taxa in $S \cup \{a,b\}$ lie below a cut-edge in $T$. Hence, $lca(c,b)$ is the root of $T[B]$ in $T$; so the edge entering $lca(c,b)$ (if it is present) can be cut.
\item \textbf{Case 3.2.3: In $T$, $c$ is sibling to the root of $T[X_j]$, and $T[S \cup \{a,b\}]$ is outgoing from $T[X_j]$.} In this case we can branch as follows: \textcolor{red}{(1) Cut off $X_1$ and $X_2$, (2) Cut off $b$}. The crucial point underpinning safety here is that it is not possible to have $b, c$ and taxa from $X_j$ together in a component. This is because $c$ is at the root of $T$, meaning that for every $\{b,c,x\}$ where $x \in X_j$, $T|\{b,c,x\} \neq T'|\{b,c,x\}$. Now, if $b$ is not in a singleton component, then it is together with $a$ in a component (putting us in branch 1), or $a$ is a singleton and $b$ is together in a component with at most one of $c$ and taxa from $X_j$ (but not both). If it is together with exactly one of them, we can separate $b$ from $c$ or $X_j$ and then bring $a$ into the component, putting us back in branch 1. If it is together with neither of them, then if necessary we can separate $c$ and $X_j$ with one cut, and then bring $a$ into the component with $b$, returning us (for the final time) to branch 1. This gives recursion $T(k)=T(k-1)+T(k-2)$ with branching factor 1.618.
\end{itemize}

\end{itemize}
\end{itemize}
\end{proof}

Even when preceded by branching rules \ref{rule:unify} and \ref{rule:twohomeoroot}, there is still a risk with the classical Whidden $(t=2)$ that the branching factor rises too much.
This can happen if the $T(k-2)$ branch, when $X_1$ and $X_2$ are cut off in $T'$, subsequently induces worst-case behaviour in branching rule \ref{rule:overlap}, which is $2T(k-1)$. The $T(k-2)$ term becomes $2T((k-2)-1) = 2T(k-3)$, giving an overall running time of $2T(k-1) + 2T(k-3)$, where as usual the $2T(k-1)$ term is from cutting off $a$ or $b$. 
This has a branching factor of~$2.3593$, which is larger than we want.

To address this, we construct a specially modified version of Whidden $(t=2)$. This modification entails that, immediately after the
$T(k-2)$ branch is created, branching rules \ref{rule:weirdoverlap} and \ref{rule:fitch} are executed.
These rules ensure that, subsequently, worst case running time for branching rule \ref{rule:overlap} is avoided.
\begin{branchrule}\label{rule:weirdoverlap} 
Let $(T,F')$ be a non-disjoint instance, such that $F'$ consists of components $A, B, C$
where:
\begin{itemize}
    \item $A$ is homeomorphic and $C$ is homeomorphic, but $B$ is not.
    \item $T[C]$ is disjoint from $T[A]$ and $T[B]$, and $T[A]$ and $T[B]$ overlap on only a single edge; write $A_1, A_2$ and $B_1,B_2$ for the bipartitions of the taxa of $A$ and $B$ respectively induced by this edge.
\end{itemize}
Then cut component $B$ into $B_1$ and $B_2$ with a single cut.
\end{branchrule}
\begin{restatable}{lemma}{weirdoverlapcorrect}
    Branching rule \ref{rule:weirdoverlap} is safe and has recurrence $T(k)=T(k-1)$.
\end{restatable}
\begin{proof}
To see why this safe, suppose some optimal agreement forest reached from this point contains
a component with taxa from both $B_1$ and $B_2$. In such an agreement forest, $A$ is necessarily split into two or more components. However, we can modify this into an agreement forest that is no larger by splitting the component that spans $B_1$ and $B_2$ into its $B_1$ and $B_2$ parts with a single cut, and introducing $A$ as a single component (which is possible because $A$ is homeomorphic and $A_1, A_2$ only overlap on a single edge of $T$), saving at least one cut. 
\end{proof}

The following is not really a branching rule, but a situation in which part of the agreement forest can independently be computed to optimality in polynomial
time, reducing the number of cuts available to the remaining part of the instance. Fitch's algorithm~\cite{Fitch71}  is a well-known algorithm in phylogenetics.
\begin{branchrule}\label{rule:fitch} 
Let $(T,F')$ be a non-disjoint instance, and let $F'_{H}$ be a subset of the homeomorphic components of $F'$ with the following property: $|F'_H| \geq 2$ and in $T$,
no component of $F'_{H}$ overlaps with a component outside $F'_{H}$. Let $X_{H}$ be the union
of the taxa in the components of $F'_{H}$. Use Fitch's algorithm to cut $T[ X_H ]$ into the smallest number of pieces, such
that each piece only contains taxa from a single element of $F'_H$. 
\end{branchrule}
\begin{restatable}{lemma}{fitchcorrect}
Branching rule \ref{rule:fitch} is safe and yields recursion $T(k)=T(k-1)$ if at least
one cut is required to transform $F'_H$ into an agreement forest.
\end{restatable}
\begin{proof}
Some background: given an assignment of discrete state labels to the taxa $X$ of a phylogenetic tree $T$, Fitch's algorithm extends in polynomial time the state labels to the interior nodes of $T$ such that the number of edges with different states at the endpoints is minimized. These edges correspond to cuts that separate taxa with distinct states into distinct components.
To use Fitch in the present context,
we introduce one state label per element of $F'_H$ and allocate each taxon in $X_H$ the state label of the component it is in.

Clearly any agreement forest refining $(T,F')$ must separate the $F'_H$ in $T$ (and on the other hand sets of taxa contained in a $F'_H$ must be homeomorphic since the $F_H$ are), hence branching rule \ref{rule:fitch} actually produces an agreement forest of minimum size for the components in $F'_H$. Also, none of the cuts implied by Fitch's algorithm to build this agreement forest for $F'_H$ can be used to build an agreement forest
for components outside $F'_H$, due to the assumption that the components of $F'_H$ only overlap with each other in $T$.  Hence, if Fitch's algorithm computes that at least one edge has to be cut
to obtain an agreement forest for $F'_H$, at most $k-1$ cuts are available for any remaining components outside $F'_H$. Thus, we obtain a branching factor of at most $T(k-1)$.
\end{proof}

We can now describe the overall algorithm for rMAF; see Figure \ref{alg:mraf}. Note that here the recursion rule in step 4 is called with $t=1$.

\begin{figure}[b]
  \caption{\textsc{Algorithm rMAF.}}
\begin{enumerate}
    \item Apply tidying-up operations
    \item Apply the first among the following branching rules to be applicable:
    \begin{enumerate}\label{it:rprerules}
    \item Whidden $(t=1)$ (1)
    \item Whidden $(t\geq 3)$ (2.206)
    \item Branching rule \ref{rule:unify} (2)
    \item Branching rule \ref{rule:twohomeoroot} (2.3028)
    \item Whidden $(t=2)$ (2.42), followed by branching rule \ref{rule:weirdoverlap} or \ref{rule:fitch} if applicable.
    \end{enumerate}
    \item Apply branching rule \ref{rule:overlap} \SK{(SPLIT)} to exhaustion (2)
    \item Apply the recursion rule \SK{(i.e. decomposition)} with $t=1$
\end{enumerate}
\label{alg:mraf}
\end{figure}

Every branching rule except for Whidden $(t=2)$ has a branching factor below our target of $\bestrootednumber$. However, we still obtain on overall branching factor of at most $\bestrootednumber$:

\begin{lemma}
\label{lem:modifiedwhidden}
Let $T$ and $T'$ be two rooted 
trees on $X$ with no common cherries. Assume that none of Whidden $(t=1)$, Whidden $(t\geq 3)$, branching rule \ref{rule:unify} or branching rule \ref{rule:twohomeoroot} apply. Then applying Whidden $(t=2)$, followed by branching rule \ref{rule:weirdoverlap} or \ref{rule:fitch} if applicable, followed  by branching rule \ref{rule:overlap} to exhaustion, followed by the recursion rule with $t=1$, has overall branching factor at most~$\bestrootednumber$.
\end{lemma}

\begin{proof}
The obstacle blocking a branching factor of at most {\bestrootednumber} is when the $T(k-2)$ branch of Whidden $(t=2)$ induces worst-case behaviour, $2T(k-1)$, in branching rule \ref{rule:overlap}. We show that this is avoided. More formally, we take as our starting point the situation created by the $T(k-2)$ branch of Whidden $(t=2)$, i.e. $T'$ has been split into $X_1, X_2$ and $S \cup \{a,b\}$, and show that from this point a recursion of at most $T(k-1) + 2T(k-2)$ is obtained. Plugging $k-2$ in here for $k$, gives $T(k-3) + 2T(k-4)$. Combined with the original $2T(k-1)$ term for cutting off $a$ and $b$, we will obtain an overall running time of $2T(k-1) + T(k-3) + 2T(k-4)$, which finally yields the desired branching factor of 2.3391. 

We can assume that tidying-up operations have been applied, so $a, b$ have been collapsed into a single taxon, because they now form a common cherry, and any singleton components have been deleted. We let $\{ A, B, C \}$ be the (at most) three components in $T'$ at this point; we will not need to know their origins as $\{ X_1, X_2, S\cup \{a,b\}\}$.
We know that if $A, B, C$ are mutually disjoint in $T$, then at least two of them are non-homeomorphic, because otherwise branching rule \ref{rule:twohomeoroot} would have applied. In this situation the recursion rule in step 4 (executed to depth 1) ensures a recursion of at most $2T(k-2)$ (see Lemma \ref{ref:easydecompose}), which is faster than the required $T(k-1) + 2T(k-2)$ bound, so we are done. Hence, from this point on we know that $A, B, C$ are not mutually disjoint in $T$. We consider a number of other cases.
\begin{enumerate}
\item \emph{$T[A], T[B], T[C]$ share a common edge $e$ of $T$}. Only one of the three components can use~$e$, so branching rule \ref{rule:overlap} will initially apply to two of the components, say $A$ and $B$. In the branch where $A$ is split, we simply assume worst-case behaviour i.e. only one cut is needed to split $A$. This contributes $T(k-1)$ to the recursion. In the branch where $B$ is split, we also assume that only a single cut is needed to split $B$, but observe that branching rule \ref{rule:overlap} then triggers again to separate $A$ from $C$, giving a running time of at most $2T((k-1)-1) = 2T(k-2)$ for that branch, and thus $T(k-1) + 2T(k-2)$ overall.

\item \emph{$T[A]$ overlaps with $T[B]$, $T[B]$ overlaps with $T[C]$, but $T[A]$ does not overlap with $T[C]$}. This is essentially identical to the previous case: when splitting (say) $B$ from $A$, the branch in which $B$ is kept intact, will then need to split $B$ from $C$.

\item \emph{$T[A]$ and $T[B]$ overlap, but neither overlaps with $T[C]$. $C$ is \emph{not} homeomorphic.} If $A$ and $B$ are homeomorphic, then branching rule \ref{rule:fitch} can be applied by taking $F'_H = \{A,B\}$.
At least one cut will be needed to separate them. Hence, at most
$k-1$ cuts are available for $C$, giving an overall branching factor of at most $T(k-1)$ and we are done. On the other hand, suppose that at least one of $A$ and $B$ is non-homeomorphic, say $A$. Branching rule \ref{rule:overlap} will trigger to split $A$ and $B$. We allocate time $T(k-1)$ to the branch in which $A$ is split. In the other branch (where
$B$ is split), we obtain a new disjoint instance in which $A$ and $C$ both survive, and both are non-homeomorphic. Then,
Lemma \ref{ref:easydecompose} gives a bound of at most $2T((k-1)-2) = 2T(k-3)$, yielding
at most $T(k-1) + 2T(k-3)$ overall.
\item \emph{$A$, $B$ and $C$ are all homeomorphic}. This is caught by branching rule \ref{rule:fitch},
by taking $F'_H = \{A,B,C\}$. This solves the $T(k-2)$ branch to optimality in polynomial time.
\end{enumerate}
Assuming none of the situations above apply implies that $A, B, C$ are not all homeomorphic, $T[A]$ overlaps with $T[B]$ but neither overlap with $T[C]$, and $C$ is homeomorphic. Hence, at least one of $A$ and $B$ is non-homeomorphic.  Recall that branching rule \ref{rule:overlap} first decides whether $A$ or $B$ should use the overlap edge, and then splits the other tree. If either tree requires 2 or more cuts to be split, we are done, because then the running time of branching rule \ref{rule:overlap} is at most $T(k-1)+2T(k-2)$ instead of $2T(k-1)$. 
So, assume that both trees require exactly 1 cut to be split. Let $A_1$ and $A_2$ be the bipartition of $A$ in $T$ induced by the overlap edge, and define $B_1, B_2$ similarly. In $T'$ there is a single edge that separates $A_1$ and $A_2$, and a single edge that separates $B_1$ from $B_2$: this is due to the assumption that both trees require only 1 cut to be split in $T'$.

Now, suppose in the branch where $B$ is kept intact
(and $A$ is split into $A_1, A_2$), at least one of $T[A_1], T[A_2]$ overlaps with $T[B]$. Here branching rule \ref{rule:overlap} is triggered for a second time, which is sufficient. A symmetrical observation holds for the branch where $A$ is kept intact.
So, we henceforth assume neither triggering happens. This means that $T[A_1], T[A_2], T[B_1]$ and $T[B_2]$ are mutually disjoint and that $T[A], T[B]$ overlap on a single edge. Continuing, suppose at least one of $A_1, A_2, B_1, B_2$ is non-homeomorphic. Without loss of generality, suppose it is $B_1$ (so $B$ is non-homeomorphic). If $A$ is homeomorphic, then branching rule \ref{rule:weirdoverlap} applies, and we are done.
Hence, we can assume that $A$ is not homeomorphic. Branching rule \ref{rule:overlap} then splits $A$ and $B$. In the branch in which $A$ is kept intact,
$A$ and $B_1$ exist as disjoint non-homeomorphic components, so Lemma \ref{ref:easydecompose}
yields $2((k-1)-2) = 2(k-3)$ for solving $A$ and $B_1$. Combined with $T(k-1)$ for the branch
in which
$B$ is kept intact, we obtain $T(k-1) +2T(k-3)$ overall.
Finally, let us assume that all of $A_1, A_2, B_1, B_2$ are homeomorphic. Here one single cut is necessary and sufficient to obtain an agreement forest for $T$ and $T'$: deleting the unique overlap edge in $T$. 
This single cut will
be discovered by subsequent branching rules after branching to at most depth 1, i.e. in polynomial time.
\end{proof}

\begin{theorem}
\label{thm:rooted}
\textsc{Algorithm rMAF} solves the problem rMAF in time $O^{*}(\bestrootednumber^k)$.
\end{theorem}

\section{Future work}
\label{sec:future}
Implementations of FPT branching algorithms for uMAF and, in particular, rMAF have been aggressively
optimized~\cite{WhiddenZB14,whidden2018calculating,yamada2019better} facilitating their use in practice. We  are hopeful that 
our branching scheme, which is conceptually straightforward, can lead to further practical speed-ups: not just because of the reduced branching factor, but also because after branching rule~\ref{rule:overlap} has been applied to exhaustion, the resulting components in $F'$ have disjoint embeddings in $T$: we obtain a set of disjoint tree-tree instances. This opens the door to applying existing kernelization algorithms to these tree-tree instances encountered deeper in the recursion (``branch and reduce''). Existing sophisticated kernelization algorithms and bounding schemes do not generalize well to tree-forest instances~\cite{van2022reflections}, so this is an interesting opportunity. \SK{An implementation should also take into account that, in order to attain state-of-the-art performance \emph{in practice}, it is usually necessary to incorporate additional speed-ups that are not part of the core branching scheme. For example, although it does not improve worst-case performance, the \emph{common cluster} reduction rule is known to give a dramatic performance boost in practice \cite{li2017computing}. Other speed-ups include fast and aggressive lower bounding schemes that can very quickly return $\qfalse$ (i.e. fail) if it is impossible to reach an agreement forest from a given point with the remaining budget of cuts. It is also possible that, in practice, some of the more exotic flanking branching rules that we have developed (such as branching rule \ref{rule:twohomeoroot}) do not need to implemented at all in order to secure the main benefits of our algorithm. Careful profiling of where the algorithm spends its time on realistic data sets is required to guide such engineering decisions. We defer such implementation issues to future research.} On the theoretical side it is natural to ask how far the branching factors obtained in this
article can be reduced, and what improvements our branching scheme might yield for other variants of rMAF
and uMAF, for example on multiple trees or nonbinary trees.

\bibliography{jcss}

\end{document}